\documentclass[10pt,a4paper]{article}

\usepackage{authblk}

\usepackage[numbers,sort&compress]{natbib}

\usepackage[a4paper,margin=1in]{geometry}
\usepackage{titlesec}

\let\coloneqq\relax

\usepackage[T1]{fontenc}

\usepackage{amsthm}
\usepackage{amssymb}
\usepackage{amsmath}
\DeclareMathOperator{\Tr}{Tr}
\usepackage{bbold}
\usepackage{bbm}
\usepackage{nonfloat}
\usepackage{xcolor}
\definecolor{myrefcolor}{rgb}{0.067,0.5,0.5}
\usepackage[
    breaklinks,
    pdftex,
    colorlinks=true,
    linkcolor=myrefcolor,
    citecolor=myrefcolor,
    urlcolor=myrefcolor
]{hyperref}
\usepackage[nameinlink,capitalise,noabbrev]{cleveref}
\crefname{thm}{Theorem}{Theorems}
\Crefname{thm}{Theorem}{Theorems}

\crefname{prop}{Proposition}{Propositions}
\Crefname{prop}{Proposition}{Propositions}

\crefname{lemma}{Lemma}{Lemmas}
\Crefname{lemma}{Lemma}{Lemmas}

\crefname{cor}{Corollary}{Corollaries}
\Crefname{cor}{Corollary}{Corollaries}

\crefname{Def}{Definition}{Definitions}
\Crefname{Def}{Definition}{Definitions}

\crefname{rem}{Remark}{Remarks}
\Crefname{rem}{Remark}{Remarks}

\crefname{ex}{Example}{Examples}
\Crefname{ex}{Example}{Examples}

\crefname{equation}{Eq.}{Eqs.}
\Crefname{equation}{Equation}{Equations}

\crefname{figure}{Fig.}{Figs.}
\Crefname{figure}{Figure}{Figures}

\crefname{section}{Sec.}{Secs.}
\Crefname{section}{Section}{Sections}
\usepackage{braket}
\usepackage{dsfont}
\usepackage{mathdots}
\usepackage{mathtools}
\usepackage{enumerate}
\usepackage[shortlabels]{enumitem}
\usepackage{csquotes}
\usepackage{stmaryrd}
\usepackage[cal=boondox]{mathalfa}
\usepackage{graphicx}
\usepackage{stackengine}
\usepackage{scalerel}
\usepackage{microtype}
\usepackage{array}
\usepackage{makecell}
\newcolumntype{x}[1]{>{\centering\arraybackslash}p{#1}}
\usepackage{tikz}
\usepackage{pgfplots}
\usetikzlibrary{shapes.geometric, shapes.misc, positioning, arrows, arrows.meta, decorations.pathreplacing, decorations.pathmorphing, patterns, angles, quotes, calc}
\usepackage{booktabs}
\usepackage{xfrac}
\usepackage{siunitx}
\usepackage{centernot}
\usepackage{comment}
\usepackage{chngcntr}
\usepackage{breakurl}
\def\clearthms#1{ \@for\tname:=#1\do{\cleartheorem\tname} }

\usepackage{thmtools}
\newtheorem*{thm*}{Theorem}

\newtheorem*{prop*}{Proposition}
\newtheorem{lemma}{Lemma}
\newtheorem*{lemma*}{Lemma}

\newtheorem*{cor*}{Corollary}

\newtheorem*{cj*}{Conjecture}

\newtheorem*{Def*}{Definition}

\newtheorem{remark}{Remark}

\makeatletter
\def\thmhead@plain#1#2#3{%
  \thmname{#1}\thmnumber{\@ifnotempty{#1}{ }\@upn{#2}}%
  \thmnote{ {\the\thm@notefont#3}}}
\let\thmhead\thmhead@plain
\makeatother

\newcommand{\bb}{\begin{equation}\begin{aligned}\hspace{0pt}}
\newcommand{\be}{\begin{equation}\begin{aligned}\hspace{0pt}}
\newcommand{\bbb}{\begin{equation*}\begin{aligned}}
\newcommand{\ee}{\end{aligned}\end{equation}}
\newcommand{\eee}{\end{aligned}\end{equation*}}
\newcommand*{\coloneqq}{\mathrel{\vcenter{\baselineskip0.5ex \lineskiplimit0pt \hbox{\scriptsize.}\hbox{\scriptsize.}}} =}

\newcommand{\ketbra}[1]{\ket{#1}\!\!\bra{#1}}

\DeclareMathAlphabet{\pazocal}{OMS}{zplm}{m}{n}

\DeclareMathOperator{\supp}{supp}

\newcommand{\lsmatrix}{\left(\begin{smallmatrix}}
\newcommand{\rsmatrix}{\end{smallmatrix}\right)}

\stackMath

\stackMath

\makeatletter
\newcommand*\rel@kern[1]{\kern#1\dimexpr\macc@kerna}
\newcommand*\widebar[1]{%
  \begingroup
  \def\mathaccent##1##2{%
    \rel@kern{0.8}%
    \overline{\rel@kern{-0.8}\macc@nucleus\rel@kern{0.2}}%
    \rel@kern{-0.2}%
  }%
  \macc@depth\@ne
  \let\math@bgroup\@empty \let\math@egroup\macc@set@skewchar
  \mathsurround\z@ \frozen@everymath{\mathgroup\macc@group\relax}%
  \macc@set@skewchar\relax
  \let\mathaccentV\macc@nested@a
  \macc@nested@a\relax111{#1}%
  \endgroup
}

\counterwithin*{equation}{part}
\counterwithin*{thm}{part}
\counterwithin*{figure}{part}

\tikzset{meter/.append style={draw, inner sep=10, rectangle, font=\vphantom{A}, minimum width=30, line width=.8, path picture={\draw[black] ([shift={(.1,.3)}]path picture bounding box.south west) to[bend left=50] ([shift={(-.1,.3)}]path picture bounding box.south east);\draw[black,-latex] ([shift={(0,.1)}]path picture bounding box.south) -- ([shift={(.3,-.1)}]path picture bounding box.north);}}}
\tikzset{roundnode/.append style={circle, draw=black, fill=gray!20, thick, minimum size=10mm}}
\tikzset{squarenode/.style={rectangle, draw=black, fill=none, thick, minimum size=10mm}}

\definecolor{Blues5seq1}{RGB}{239,243,255}
\definecolor{Blues5seq2}{RGB}{189,215,231}
\definecolor{Blues5seq3}{RGB}{107,174,214}
\definecolor{Blues5seq4}{RGB}{49,130,189}
\definecolor{Blues5seq5}{RGB}{8,81,156}

\definecolor{Greens5seq1}{RGB}{237,248,233}
\definecolor{Greens5seq2}{RGB}{186,228,179}
\definecolor{Greens5seq3}{RGB}{116,196,118}
\definecolor{Greens5seq4}{RGB}{49,163,84}
\definecolor{Greens5seq5}{RGB}{0,109,44}

\definecolor{Reds5seq1}{RGB}{254,229,217}
\definecolor{Reds5seq2}{RGB}{252,174,145}
\definecolor{Reds5seq3}{RGB}{251,106,74}
\definecolor{Reds5seq4}{RGB}{222,45,38}
\definecolor{Reds5seq5}{RGB}{165,15,21}

\usepackage{pgfplots}
\newtheorem{theorem}{Theorem}

\newtheorem{definition}{Definition}

\usepackage{amsmath}
\usepackage{algpseudocode}
\usepackage{algorithm}

\definecolor{antonio}{rgb}{.2,.5,.1}

\renewcommand{\ketbra}[2]{\ket{#1}\!\bra{#2}}
\usepackage{libertinus}

\allowdisplaybreaks

\begin{document}
\title{Efficient learning of Clifford disentanglers and typical\\ $t$-doped unitaries with exponentially more $T$ gates}
\author[1]{Gerard Aguilar}
\author[1]{Sofiene Jerbi}
\author[1]{Jens Eisert}
\author[1]{Lennart Bittel}
\affil[1]{\normalsize Dahlem Center for Complex Quantum Systems, Freie Universität Berlin,\protect\\ Arnimallee 14, 14195 Berlin, Germany}
\date{\vspace{-0.5em}\today}
\maketitle
\begin{abstract}
Highly entangled and highly non-stabilizer quantum states need not be hard to learn. We give efficient algorithms for testing and recovering hidden tensor-product structure in unknown pure state vectors of the form $\ket{\psi}=U_C\bigotimes_i\ket{\psi_i}$, where $U_C$ is an arbitrary unknown Clifford unitary. Although the Clifford can thoroughly scramble the visible product structure, we prove that the Bell distribution retains a characteristic family of quadratic symmetries. By simultaneously block-diagonalising these symmetries, our algorithms linearize the problem and manage to recover both a disentangling Clifford and the hidden partitions with polynomial sample and computational complexity. This may be viewed as an extension of the abelian StateHSP paradigm in which classical post-processing exposes genuinely quadratic structure. Applied to Choi states, the method yields efficient proper learning algorithms for typical $t$-doped Clifford unitaries in regimes containing exponentially more $T$ gates than previously accessible: the required condition fails only for an exponentially small fraction of circuits when $t\sim n$, and continues to hold for a constant fraction even when $t=2n$. Our framework also provides tools for compressing structured many-body Hamiltonians and suggests benchmarking protocols for encoded logical product states in the early fault-tolerant regime.

\end{abstract}

\tableofcontents

\section{Introduction}

Learning full descriptions or selected properties of quantum objects, such as quantum states and channels, lies at the intersection of two central research directions in quantum information. From the perspective of certification and benchmarking, learning provides scalable tools for verifying the correct functioning of quantum devices and characterising the states and processes they prepare \cite{kliesch2021theory,eisert2020quantum,gheorghiu2019verification}. From the perspective of computational learning theory, it asks which information contained in a quantum system can be extracted efficiently by a computationally bounded observer. This viewpoint helps identify the sources of quantum complexity and has fundamental consequences for our understanding of pseudorandomness and cryptography in quantum systems \cite{gheorghiu2019verification,ananth2022cryptography,ji2018pseudorandom}.

 In full generality, however, this task becomes intractable: quantum state tomography or process tomography require resources that scale exponentially in system size \cite{cramer2010efficient,anshu2024survey}. As a result, much of quantum learning theory is concerned with identifying structured classes of states and unitaries for which efficient characterisation becomes possible \cite{cramer2010efficient,leone2024learningstates,montanaro2017learning,mele2025efficient,chen2026optimal,bouland2025state,hinsche2026abelian,arunachalam2022optimal}. Among these, some well-known families of learnable quantum states include low-magic \cite{leone2024learningstates} and low-entanglement states \cite{cramer2010efficient}. These are both characterised by having an efficient structured representation: low-magic states admit sparse descriptions in the Pauli basis, while low entanglement states admit low-rank tensor network descriptions. Moreover, in the limit case, namely for stabilizers and product states, both structures induce symmetries in the state which allow to map the learning problem to an \textit{abelian} hidden subgroup problem for efficient recoverability \cite{hinsche2026abelian}.

This paradigm, however, leaves out highly structured families of states that admit succinct descriptions while exhibiting nearly maximal entanglement and extensive non-stabilizerness. This raises a question closely connected to quantum pseudorandomness: when does efficiently generated structure survive scrambling and remain accessible to a computationally bounded observer? Pseudorandom quantum states are efficiently generated states that nevertheless remain computationally indistinguishable from Haar-random states, even given polynomially many copies \cite{ji2018pseudorandom,ananth2022cryptography}. We will show that previously assumed relevant structures such as entanglement, magic, or the absence of an immediately visible tensor-product structure are not by themselves sufficient indications of pseudorandomness. What matters is instead whether the state retains efficiently detectable invariants from which its hidden structure can be reconstructed.

In particular, in this work we will look at the problem of testing and learning state vectors of the form $\ket{\psi}=U_C\bigotimes_i\ket{\psi_i}_{c_i}$, for potentially arbitrary $\ket{\psi_i}$. Such states can be highly resourceful under any common measure, and in particular can showcase both high entanglement and high magic well beyond the classically simulable or accessible regime. We will show that they are, however, highly symmetric and that retrieving these symmetries can also be understood as an instance of an \textit{abelian} hidden subgroup problem. We will then show that this can be efficiently solved by performing Bell sampling and appropriate classical post-processing of their Bell distribution. Bell sampling has already proven to be a useful primitive in quantum learning theory, where it has been used, for instance, to perform fidelity estimation, learn unknown stabilizer states, or to learn hidden product structures \cite{montanaro2017learning,hangleiter2024bell,hinsche2026abelian}. It has also recently proved useful in studying quantum weight and shadow enumerators \cite{miller2026experimental}, and further extensions of it, such as the Bell-difference distribution, have been shown to play a central role in tasks such as stabilizer testing \cite{gross2021schur,grewal2024improved,grewal2025efficient}. Having recovered these symmetries, the genuinely quadratic behaviour of the problem requires a more involved analysis of the distribution's properties than in the standard hidden-cut problem. We will show that, through appropriate classical post-processing of the samples, one is able to efficiently recover a disentangling Clifford $U_C$, a partition $\{c_i\}_i$, and even a full description of the state whenever the $\{\psi_i\}_i$ themselves are amenable to efficient learning.

A family of unitaries that can lie in this highly magic, but also highly structured regime are $T$-doped unitaries with up to $2n$ many $T$ gates. Indeed, we show that in many cases, such unitaries have Choi states of the form we study here, and can therefore be efficiently learnt. Learning $T$-doped unitaries as well as states generated by them has been an object of intense study in recent years \cite{leone2024learning,leone2024learningstates,grewal2025efficient,grewal2024improved,chia2024efficient,chen2025stabilizer,arunachalam2026learning} as they arise as a natural circuit class in the near-term fault-tolerant setting \cite{bravyi2005universal,o2017quantum,campbell2017roads}, provide a controlled ansatz for studying the onset of
chaotic behavior in many-body dynamics \cite{zhou2020single,bejan2024dynamical,gu2024doped}, and they can also be used to generate high order approximate unitary designs at a small cost of magic \cite{Homeopathy}.
Efficient learning algorithms have nevertheless always been restricted to low resource regimes with $O(\log(n))$ many $T$ gates. We show here, however, that in a regime with exponentially more $T$ gates, namely for $\sim n$ many, only an exponentially small fraction of these circuits cannot be efficiently learnt with our algorithm. Even at a higher regime with $\sim 2n$ still a constant fraction of unitaries can be learnt. Besides theoretical interest, we believe that these results bear significant relevance for practical tasks such as benchmarking error corrected settings or 
compression of many-body Hamiltonians and states.

\subsection{Setting}

As stated, our goal can be phrased as finding efficient methods to construct exact Clifford disentanglers for arbitrary states. That is, given some generic unknown state vector $\ket{\psi}$, find a description of the form
\begin{equation}
    \label{eqn:clifford_dressed}\ket\psi=U_C\left(\bigotimes_{i=1}^m\ket{\psi_i}_{C_i}\right),
\end{equation}
in case it exists, where $C_1\sqcup\dots\sqcup C_m=[n]$ is some partition of the physical qubits and $U_C$ 
some arbitrary Clifford unitary. To that end, we will define two independent tasks. First, a testing problem, such that given copies of the unknown state vector $\ket{\psi}$, we can determine whether it can be written as in \Cref{eqn:clifford_dressed}, and then the corresponding learning task, where the goal is to recover this description. We now give a formal definition for both tasks.

\begin{definition}[{[Clifford-scrambled hidden-cut testing problem]}]
    \label{def:testing_task}
    Let $\delta,\epsilon> 0$, $0<\eta<1$. Let $\ket{\psi}$ be an unknown pure state on $n>2$ qubits. Suppose we are promised that $\ket{\psi}$ falls into one of two situations:
    \begin{itemize}
        \item \textbf{Case A:} 
        \begin{align}
            \exists U_C,\exists A\subset[n], B =[n]\backslash A : \ket{\psi}=U_C\bigl(\ket{\psi_A}\otimes \ket{\psi_B}\bigr)
        \end{align}
        moreover, for any $P\in\mathcal P_n$, the state satisfies one of the two conditions
        \begin{equation}
            |\langle P\rangle_{\psi}|=1\quad\text{or}\quad|\langle P\rangle_{\psi}|\leq 1-\eta.
        \end{equation}
        \item \textbf{Case B:} 
        \begin{align}
            \forall U_C, \forall A\subset[n], B =[n]\backslash A: \left\|\ket{\psi}-U_C\bigl(\ket{\psi_A}\otimes \ket{\psi_B}\bigr)\right\|_2\geq \epsilon,
        \end{align}
        where we express the distance between the states using the 2-norm on their state vectors after minimizing over global phases (see \Cref{remark:distance_measure})
        \begin{equation}
           \min_{\theta\in[0,2\pi)}\left\|\ket{\psi}-e^{i\theta}\ket{\phi}\right\|_2
            =\sqrt{2-2|\langle\psi|\phi\rangle|}.
        \end{equation}
    \end{itemize}
    The problem is then to determine whether the state is in Case A or in Case B with probability at least $1-\delta$ given copies of $\ket{\psi}$.
\end{definition}

For the testing task above, we choose to define it for a single cut $A|B$, but it can trivially be reformulated for a more generic case with multiple cuts. In fact, for the learning problem, we will want that the algorithm finds the finest description possible, and hence we will define it for the general multi-cut setting.

\begin{definition}[{[Clifford-scrambled hidden-cut learning problem]}]
    \label{def:learning_task}
    Let $\delta,\epsilon>0$, $0<\eta<1$. Let $\ket\psi$ be a pure state on $n>2$ qubits and $m\in[2,n]$ fixed . Suppose that $\ket{\psi}$ is of the form
    \begin{equation}
        \ket{\psi}=U_C\left(\bigotimes_{i=1}^m\ket{\psi_i}_{c_i}\right)
    \end{equation}
    for some partition $c_1\sqcup\dots \sqcup c_m=[n]$, and moreover it satisfies that for any $P\in\mathcal P_n$ 
    \begin{equation}
            |\langle P\rangle_{\psi}|=1\quad\text{or}\quad|\langle P\rangle_{\psi}|\leq 1-\eta.
    \end{equation}
    Then the learning problem asks, given copies of $\ket{\psi}$ to find a set partition $\tilde c_1\sqcup \dots\sqcup \tilde c_{m'}$, as well as a Clifford $U_{\widetilde C}$ such that 
    \begin{equation}
       \left\|U_{\widetilde C}\left(\bigotimes_{i=1}^{m'}\ket{\tilde\psi_i}_{\tilde c_i}\right)-\ket\psi\right\|_2\leq \epsilon.
    \end{equation}
    for $m'\geq m$ with probability at least $1-\delta$.
\end{definition}

In this work, we show that these two tasks can be solved efficiently both in the number of samples we require, as well as in runtime. 

\begin{remark}[{[Stabilizer-gap promise]}]
    Note that both in \Cref{def:testing_task,def:learning_task}, we demand a promise gap for the presence of stabilizers in the unknown state vector $\ket{\psi}.$ We discuss this requirement later on in \Cref{rem:stab_gap}.
\end{remark}

\subsection{Our results}

The main result of this work is to give efficient runtime algorithms that solve the problems in \Cref{def:testing_task,def:learning_task} given a polynomial number of copies of the state. Our algorithms will rely on Bell sampling. The main difficulty to overcome in our setting is that the action of the unknown Clifford seemingly scrambles the Bell support of the state almost arbitrarily, and hence washes off any kind of structure in the distribution.

The key point of our work is to show that it is not the case, and that state vectors 
of the form $C\bigotimes_i\ket{\psi_i}$ remain structured in a very specific sense.
In order to show this, we introduce a framework to study quadratic symmetries of the Bell distribution of a state. That is, we look at all structures of a state $\psi$ that can be expressed by 
\begin{equation}\label{eq:Bell-symmetry}
    q^TMq=c,\quad\forall q\in\supp(\mathsf b_\rho)
\end{equation}
where $\mathsf b_\rho$ is the Bell distribution of $\psi$, $q$ is the symplectic representation of the Paulis that label the corresponding Bell states, $c\in\mathbb F_2$, and $M\in M_{2n}(\mathbb F_2)$. These structures then give rise to some generalised purity equations, of the 
form
\begin{equation}
    1=\sum_{q}\mathsf b_\rho(q)(-1)^{q^TMq+c}.
\end{equation}

In this framework, standard subsystem purity equations for product state vectors $\ket{\psi_A}\otimes\ket{\psi_B}$ 
can be seen as arising from Bell symmetries 
\begin{equation}
    \label{eqn:product_symm}
    M_A=\Omega_+^A=\bigoplus_{i\in A}\begin{pmatrix}
        0 & 1\\ 0&0
    \end{pmatrix}_i,\quad M_B=\Omega_+^B=\bigoplus_{i\in B}\begin{pmatrix}
        0 & 1\\ 0&0
    \end{pmatrix}_i,
\end{equation}
which coincides with the fact 
that purity shows up in the Bell distribution as the impossibility of sampling an odd number of singlet states \cite{miller2026experimental}. We then show that the action of an arbitrary Clifford on such a state simply transforms these constraints into other quadratic symmetries via some affine 
transformation.

\begin{restatable}[{[Clifford action on quadratic Bell symmetries]}]{lemma}{mainlemma}
    \label{lemma:clifford_effect_preliminaries}
    Given some state $\rho$ such that its Bell distribution satisfies a quadratic symmetry $M$ in the form of \Cref{eq:Bell-symmetry}. Acting on $\rho$ with some Clifford $U_C$ maps $M$ to another quadratic symmetry $M'=C^{-T}(M+\lambda_C)C^{-1}$, where $C$ represents the symplectic action of $U_C$ and $\lambda_C$ consists of an affine shift depending only on $U_C$. In particular, state vectors of the form 
    $\ket{\psi}=U_C\ket{\psi_A}\otimes\ket{\psi_B}$ satisfy symmetries $M'_A=C^{-T}(\Omega_+^A+\lambda_C)C^{-1}$, $M'_B=C^{-T}(\Omega_+^B+\lambda_C)C^{-1}$
\end{restatable}

While the states that we are interested in are then still highly structured, recovering unknown quadratic symmetries is not necessarily an easy task. For standard product states, for example, even if the symmetries in \Cref{eqn:product_symm} can be written as quadratic forms, they essentially behave linearly due to their block structure, which makes the problem significantly easier. The action of $C$ on $\Omega_+^A,\Omega_+^B$ can break this linear behaviour, mapping them to some unknown $M_A,M_B$, containing genuinely quadratic terms coupling the two subsystems. Thus in principle recovering them could require searching the full space of global quadratic forms. In order to solve this, we lift the Bell distribution to a symmetric subspace $V_\rho=\{q^{\otimes 2}|q\in\supp(\mathsf b_\rho)\}\subseteq\mathbb F_2^{2n}\otimes\mathbb F_2^{2n}$, where the quadratic constraint can be linearized
\begin{equation}
     q^TMq=0\iff \langle\langle M|q^{\otimes 2}\rangle=0,
\end{equation}
where $|M\rangle\rangle$ is the vectorization of $M.$ Then, solving a system of linear equations allows us to recover a basis for
\begin{equation}
    \text{Ker}_2(\rho)
    =
    \left\{
        M\in\mathbb F_2^{2n\times 2n}
        :
        q^T M q=c_M\quad
        \forall
        q\in\supp(\mathsf b_\rho)
    \right\},
\end{equation}
that is, the vector space of all quadratic symmetries of the state. 

While a naive algorithm could run a brute force search within $\text{Ker}_2(\rho)$ in order to verify the presence of some symmetry $M$ of the form in \Cref{lemma:clifford_effect_preliminaries}, that would take exponential runtime. In order to overcome this, we can exploit some further structure of the Bell distribution of our class of states, beyond the fact that they satisfy some specific symmetries. To state this, we first prove a statement about $\text{Ker}_2(\psi_A\otimes\psi_B)$, which immediately leads to a similar structure for state vectors of the form $U_C\ket{\psi_A}\otimes\ket{\psi_B}.$

\begin{figure*}[t]
       \centering
       \includegraphics[scale=0.23]{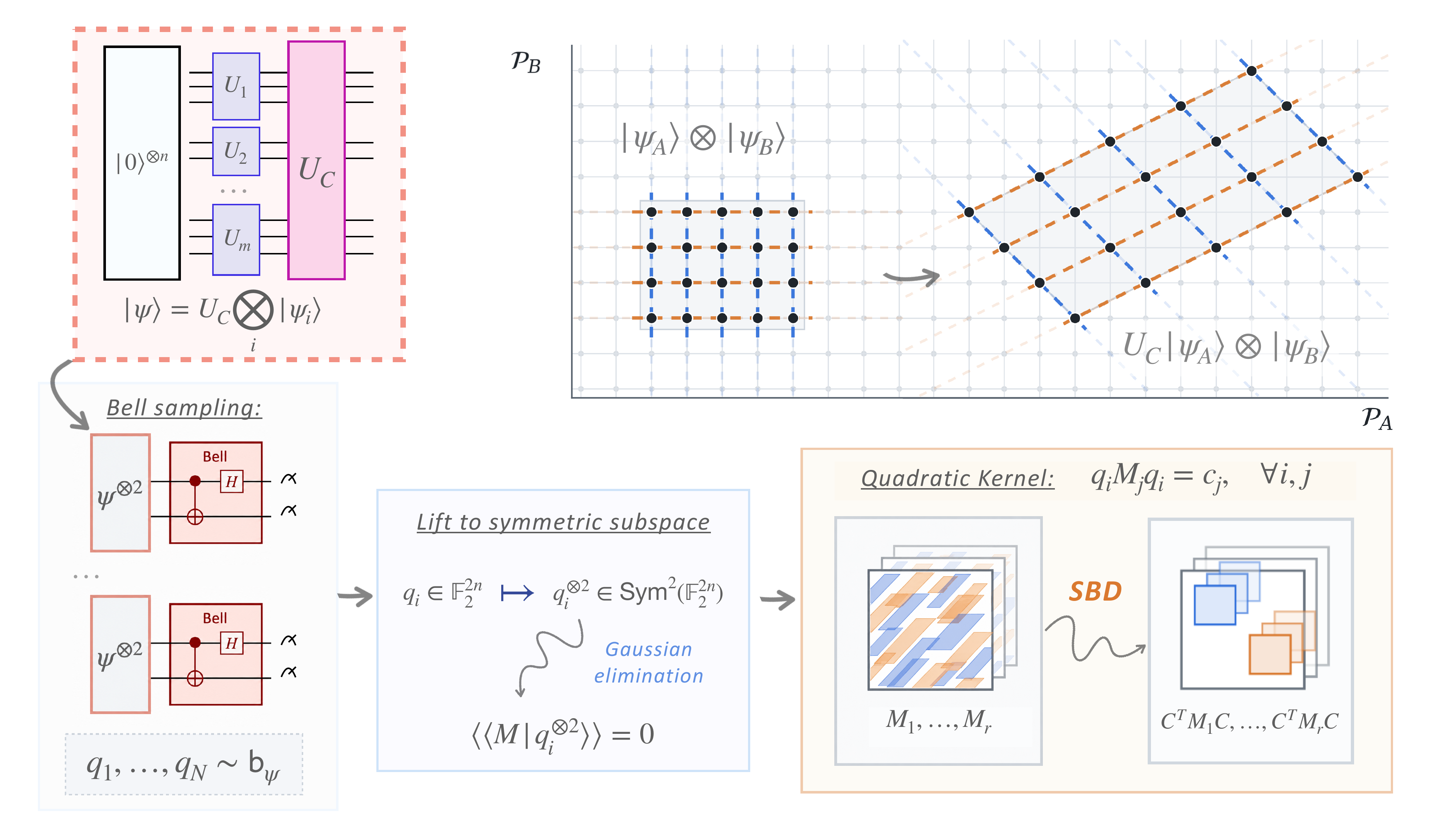}
       \caption{Graphical overview of the protocol in \Cref{alg:hidden-clifford-test}. From $m$ Bell samples $\{q_i\}_i$ of an unknown state vector $\ket{\psi}$, Gaussian elimination on the lifted vectors $q_i^{\otimes 2}\in\mathbb F_2^{2n}\otimes\mathbb F_2^{2n}$, yields a basis of the quadratic kernel $\text{Ker}_{2,m}(\psi)$ of the distribution, consisting of quadratic forms $M$ satisfying $q_i^TMq_i=c$ for every observed outcome. For state vectors of the form $\ket{\psi}=U_C\bigotimes_{i=1}^m\ket{\psi_i}_{c_i}$ this kernel has a common hidden block structure. Indeed, before the action of $U_C$, the Bell distribution factorizes across the product decomposition, forcing its quadratic symmetries to respect the corresponding blocks. The unknown Clifford scrambles both the distribution and its symmetries through an affine transformation, but preserves the common block structure in the associated Clifford frame. Performing \emph{simultaneous block-diagonalisation} (SBD) on the basis of $\text{Ker}_{2,m}(\psi)$ reveals the hidden partition and a Clifford transformation that disentangles the state.
       }
        \label{fig:schematic_plot}
   \end{figure*}

\begin{restatable}[{[Block structure of quadratic Bell symmetries]}]{theorem}{productblock}
    \label{thm:product_block}
    Let $\ket{\psi}=\ket{\psi_A}\otimes\ket{\psi_B}$ be a product state across the cut $A|B$. Assume that the state has no stabilizer. Then every quadratic Bell symmetry $M\in\text{Ker}_2(\psi)$ is block-diagonal across the cut $A|B$. Moreover, the two blocks define local quadratic symmetries on the two factors $\ket{\psi_A},\ket{\psi_B}.$
\end{restatable}

\begin{restatable}[{[Clifford-scrambled block diagonalisation]}]{cor}{blockdiagonalisable}
     \label{cor:block-diagonalisable}
    Let $\ket{\psi}=U_C\ket{\psi_A}\ket{\psi_B}$ with no stabilizers. Then every quadratic Bell symmetry $M\in\text{Ker}_2(\psi)$ can be simultaneously block-diagonalised across the cut $A|B$ via congruence with the symplectic transformation given by $C$. 
\end{restatable}

In other words, this property states that the purity of subsystems $A$ and $B$ does not only mean that $\Omega_+^A,\Omega_+^B\in\text{Ker}_2(\psi_A\otimes\psi_B)$, but also that up to the presence of stabilizers, all other symmetries of the state necessarily arise from symmetries of the subsystems. At the level of our Clifford dressed states, this then means that in some basis, the Bell distribution still factorises and this property is preserved.

Importantly, we can always assume the no-stabilizer condition to be satisfied, since our algorithm runs a stabilizer learning protocol on the Bell samples first and finds a Clifford $U_{C'}$ such that
\begin{equation}
    U_{C'}\ket{\psi}=\ket{\psi'}\otimes\ket{0}^{\otimes k},
\end{equation}
for $\ket{\psi'}$ with no stabilizers.

Our computationally efficient algorithm exploits these properties and proceeds by performing \emph{simultaneous block-diagonalisation} (SBD) of $\text{Ker}_2(\psi)$ in order to find potential structure arising from pure subsystems. We show that, by treating the symmetries as linear operators, we can phrase the SBD problem in terms of a module decomposition, allowing us to use tools from finite algebras in order to tackle this \cite{chistov1997polynomial,ronyai1990computing,wilson2008finding}. 

We then first show that for some $\ket{\psi}=U_C\ket{\psi_A}\otimes\ket
\psi_B$, we do not need to find the finest SBD of $\text{Ker}_2(\psi)$ in order to recover $A,B,U_C.$

\begin{lemma}[{[Efficient simultaneous block diagonalisation} (Informal version of \Cref{lemma:bilinear_to_linear,lemma:simple_comp,lemma:efficient_SBD})]
    \label{lemma:SBD_summary}
    Given $\text{Ker}_2(\psi)=\text{span}_{\mathbb F_2}\{M_i\}_{i=1}^{d_{\text{Ker}}}$ for some state vector $\ket{\psi}$. We can define $T_i:=\Omega(M_i+M_i^T)$, and consider the algebra
    \begin{equation}
        \mathfrak t:=\text{span}_{\mathbb F_2}\langle T_i\rangle_{i=1}^{d_{\text{Ker}}}.
    \end{equation}
    Then,
    \begin{equation}
        \mathfrak e :=\{\phi\in\text{End}_{\mathbb F_2}(\mathbb F_2^{2n}):\phi a=a\phi,\forall a\in\mathfrak t\}
    \end{equation}
    is an associative algebra that uniquely splits into some nilpotent radical $R$ and some simple components $S_i$
    \begin{equation}
        \mathfrak e=S_1\oplus\dots\oplus S_m\oplus R.
    \end{equation}
    Each component $S_i$ defines some $\mathfrak t$-invariant subspaces $V_i\subseteq \mathbb F_2^{2n}$, and in the basis defined by them, $\mathfrak t$ is block-diagonal. Moreover, if the change of basis is symplectic, then $\text{Ker}_2(\psi)$ is also block-diagonal. While this is not necessarily the finest SBD of $\text{Ker}_2(\psi)$, if the state is of the form $\ket{\psi}=U_C\ket{\psi_A}\otimes\ket{\psi_B}$, then one of the block-diagonal cuts found in this way must correspond to the partition $A|B.$ Finally, all of this can be computed in time $\text{poly}(n).$
\end{lemma}

The algorithm in \Cref{lemma:SBD_summary} in general will return a transformation $A\in GL(2n,\mathbb F_2)$ performing the SBD, not necessarily a symplectic one. We will show that since the decomposition is unique, given this $A$, it is then easy to find a symplectic change of basis $C\in \text{Sp}(2n,\mathbb F_2)$. Moreover, for a state vector $\ket{\psi}=U_C\ket{\psi_A}\otimes\ket{\psi_B}$, any such $C$ not only block-diagonalises $\text{Ker}_2(\psi)$, but also forces $\Omega_+^A,\Omega_+^B$ to be symmetries of the block-diagonalised kernel, and hence gives us a disentangling Clifford.
We at this point subsume all of this in the following theorem.
\begin{restatable}[{[Polynomial-time recovery of hidden Clifford-product structure]}]{theorem}{blockdiagonalisesym}
    \label{thm:block_diagonalisation_product}
    Given a basis $\{M_i\}_{i=1}^{d_{\text{Ker}}}$ of $\text{Ker}_2(\psi)$, for some $n$-qubit pure state vector $\ket{\psi}$, there exists an algorithm running in time $\text{poly}(n)$ that finds, if it exists, a symplectic operation $\tilde C\in\text{Sp}(2n,\mathbb F_2)$ that simultaneously block-diagonalises $\{M_i\}_{i=1}^{d_{\text{Ker}}}$. For a state of the form $\ket{\psi}=U_C\bigotimes_{i=1}^m\ket{\phi_i}_{c_i}$, satisfying the conditions in \Cref{def:learning_task}, this procedure finds blocks corresponding to each $c_i$, and a Clifford unitary $U_{\tilde C}^{\dagger}$ which necessarily disentangles $\ket\psi$ across the $c_i$'s.
\end{restatable}

We can then finally show that the number of copies of $\ket{\psi}$ required to solve tasks in \Cref{def:testing_task,def:learning_task} by an algorithm relying on the techniques above is polynomially bounded, and thus our algorithms are fully efficient.

\begin{restatable}[{[Efficient testing]}]{theorem}{maintesting}
    \label{thm:testing_alg}
    Let $\ket{\psi}$ be a state that is of the form $U_C\ket{\psi_A}\otimes\ket{\psi_B}$ as in \Cref{def:testing_task} (A) or $\epsilon$-far from it in $2$-norm for any choice of $U_C,\psi_A,\psi_B$ (B). Then, there is an algorithm running in time $\text{poly}(n,\frac{1}{\lambda},\log\frac{1}{\delta})$, such that given access to $N=O(\frac{1}{\lambda^2}(n^2+\log(1/\delta)))$ copies of the state, for $\lambda=\min(\eta,\epsilon)$ determines if $\ket{\psi}$ is in Case A or Case B with probability at least $1-\delta.$
\end{restatable}

\begin{restatable}[{[Efficient learning]}]{theorem}{mainlearning}
    \label{thm:learning_alg}
    Let $\ket{\psi}=U_C\left(\bigotimes_{i=1}^m\ket{\psi_i}_{c_i}\right)$ satisfying the stabilizer condition in \Cref{def:learning_task}. Then, there is an algorithm running in time $\text{poly}(n,\frac{1}{\lambda},\log\left(\frac{1}{\delta}\right))$, for $\lambda=\min(\eta,\epsilon/\sqrt{n})$, such that given access to $N=O\left(\frac{1}{\lambda^2}(n^2+\log\left(\frac{1}{\delta}\right))\right)$ many copies of $\ket{\psi}$, outputs a Clifford unitary $U_{\tilde C}$ and a set partition $\{\tilde c_i\}_{i=1}^{m'}$ ($m'\geq m$) satisfying $||\ket{\psi}-U_{\tilde C}\bigotimes_{i=1}^{m'}\ket{\phi_i}_{\tilde c_i}||_2\leq \epsilon$, for some choice of $\{\phi_i\}_{i=1}^{m'}$, with probability at least $1-\delta.$
    
\end{restatable}

The sample analysis of \Cref{thm:testing_alg} and \Cref{thm:learning_alg} is split in two main parts. First, we show that $\text{poly}(n,\frac{1}{\epsilon},\log\frac{1}{\delta})$ many samples are enough for an exhaustive search algorithm to succeed. Indeed, for that number of samples, only symmetries $M=C^{-T}(\Omega_+^S+\lambda_C)C^{-1}$ incurring an error at most $\epsilon$, are observed in $\text{Ker}_{2,N}(\psi)$ (the subindex $N$ distinguishes the finite sample kernel from the ideal one), and hence an algorithm that simply performs brute-force search succeeds with probability $1-\delta.$ In a subsequent step, we show that the properties of \Cref{cor:block-diagonalisable} also hold with high probability with a bounded number of samples. These are then enough to prove the overall efficiency of our algorithm, which we describe in detail in \Cref{alg:hidden-clifford-test}.

Note that the algorithm in \Cref{thm:learning_alg} only learns a description of the disentangling Clifford, and the partitions, however this can easily be lifted to a full state learning algorithm by using results from \cite{bakshi2025learning}. 

\begin{restatable}[{[Full state learning]}]{cor}{learningstates}
    \label{cor:state_learning}
    Let $\ket{\psi}$ be of the form in \Cref{def:learning_task}, with $|c_i|=\log_2d,\forall i$. Then, for any $\epsilon \in (0,1/\sqrt{3})$, using \Cref{alg:hidden-clifford-test} and Algorithm 4.19 in Ref.~\cite{bakshi2025learning}, a full description $\ket{\phi}$ such that
    \begin{equation}
        ||\ket{\psi}-\ket{\phi}||_2\leq\epsilon
    \end{equation}
    can be learnt with success probability at least $1-\delta$ with a sample complexity of $N=O\big(\frac{1}{\lambda^2}(n^2+\frac{m'd}{\lambda^2}\log\frac{1}{\delta})+m'd^2\log\frac{m'}{\delta}\big)$ and in time $\text{poly}(n,d,\frac{1}{\lambda},\log\frac{1}{\delta})$, for $\lambda=\min(\eta,\epsilon/\sqrt{n})$.
\end{restatable}
Note that in \Cref{cor:state_learning} we assume all the partitions to be of the same size for simplicity, but the result holds in general. This then not only allows us to learn states of this form, but also certain unitaries given access to their Choi state. Indeed, given any unitary of the form
\begin{equation}
    \label{eqn:conjugated_product_U}
    U=U_{C_1}\bigg(\bigotimes_iU_i\bigg)U_{C_2},
\end{equation}
its Choi state has the form in \Cref{def:learning_task}, since
\begin{equation}
    \ket{\Phi_U}=U_{C_1}\bigg(\bigotimes_iU_i\bigg)\otimes U_{C_2}^T\ket{\Phi_0}=(U_{C_1}\otimes U_{C_2}^T)\bigotimes_i\ket{\Phi_{U_i}},
\end{equation}
and thus, for the right $\ket{\Phi_{U_i}}$, $\ket{\Phi_U}$ can be learnt efficiently. The ability to learn such Choi states allows us to use some disentangling techniques introduced in Ref.~\cite{fux2025disentangling}, in order to show some striking implications of our algorithms for $T$-doped Clifford unitaries.  

\begin{restatable}[{[Learning from algebraically independent Pauli rotations} (Informal version of \Cref{lemma:disentangling} and \Cref{lemma:properlearning})]{theorem}{independentpaulis}
    \label{thm:independent_pauli_learning}
Given a $k$-doped Clifford unitary $U=U_{C_1}\cdot T\cdot\dots \cdot U_{C_k}\cdot T\cdot U_{C_{k+1}}$, with $k\leq 2n$, this can always be written as $U=U_C\prod_{i=1}^k e^{i\pi/8P_i}$. Then if the Paulis $\{P_i\}_{i=1}^k$ are algebraically independent, then the Choi state of $U$ has the form in \Cref{def:learning_task}, and a proper description of $U$ in terms of $U_C$ and $\{P_i\}_{i=1}^k$ is efficiently learnable using \Cref{alg:hidden-clifford-test} and \Cref{cor:state_learning}. Moreover, this still works arbitrary Pauli doping $\{e^{i\theta_iP_i}\}_{i=1}^k.$
\end{restatable}

Indeed, we show that some $T$-doped unitaries, in spite of exhibiting large amounts of entanglement or magic, still have a Choi state retaining a hidden Clifford-product structure, and hence remain efficiently learnable. Importantly, our proper learner does not rely solely on a global approximation of $\ket{\Phi_U}$, but rather performs state tomography on single-qubit factors, thereby providing direct control of the induced channel error. Under the exact $T$-doped promise, these factors belong to a finite, constantly separated family, so a proper description of $U$ can be identified exactly with high probability. More generally, for continuously parametrized Pauli rotations, the angle-estimation errors accumulate only linearly, so estimating each factor to accuracy $\epsilon/k$ yields an $\epsilon$ approximation in diamond norm.

Finally, we show that the condition that the Paulis are algebraically independent, and hence this structure, is maintained with high probability over the choice of Cliffords, at regimes well beyond the standard $O(\log(n))$-doped regime.
\begin{restatable}[{[Probability of algebraic independence]}]{lemma}{independentpauliprobabilities}
    For a $k$-doped Clifford unitary $U=U_{C_1}\cdot T\cdot\dots \cdot U_{C_k}\cdot T\cdot U_{C_{k+1}}$, with $U_{C_i}$ chosen uniformly at random from the Clifford group, the Pauli operators $\{P_i\}_{i=1}^k$ after conjugating the $T$ gates through, are algebraically independent with probability 
    \begin{equation}
        \label{eqn:success_p}
        p=\prod_{i=0}^{k-1}\left(1-\frac{2^i-1}{4^n-1}\right).
    \end{equation}
\end{restatable}

\begin{restatable}[{[Typical learnability of $T$-doped Clifford circuits]}]{cor}{probabilities}
    \label{cor:t_doped_probability}
    Given $k-$doped Clifford unitaries of the form $U=U_{C_1}\cdot T\cdot\dots \cdot U_{C_k}\cdot T\cdot U_{C_{k+1}}$, with $U_{C_i}$ chosen uniformly at random from the Clifford group, they can be efficiently learnt using the algorithm in \Cref{cor:state_learning} with:
    \begin{itemize}
    \item $1-O(\frac{1}{2^{c n}})$ success probability, if  $k\leq 2n-c n$ with $c>0$.
    \item $1-O(\frac{1}{n^c})$ success probability, if $k\leq 2n-c\log n$ with $c>0$.
    \item $\Theta(1)$ success probability, if $k\leq 2n-c$ with $c\geq0$.
\end{itemize}
\end{restatable}

\begin{algorithm}[t]
\caption{Computationally efficient algorithm for task in \Cref{def:learning_task}}
\label{alg:hidden-clifford-test}
\begin{algorithmic}[1]

\Statex \textbf{Input:} 
\(N=O(\frac{1}{\lambda^2}(n^2+\log(1/\delta)))\), for $\lambda=\min(\eta,\epsilon/\sqrt{n})$ copies of some state vector $\ket{\psi}.$ 

\Statex \textbf{Promise:} 
There exist a Clifford $U_C$, and some partition $\{c_i\}_{i=1}^m$ such that for some $\{\psi_i\}_{i=1}^m$ it holds that $\ket{\psi}=U_C\bigotimes_{i=1}^m\ket{\psi_i}_{c_i}$. Moreover, $\forall P\in\mathcal P_n$ either $|\langle P\rangle_\psi|=1$ or $|\langle P\rangle_\psi|\leq1-\eta$.
\Statex \textbf{Output:} 
A Clifford $U_{\tilde C}$ and a partition $\{\tilde{c}_i\}_{i=1}^{m'}$ for $m'\geq m$ such that $||U_{\tilde C}\bigotimes_{i=1}^{m'}\ket{\phi_i}_{\tilde{c}_i}-\ket{\psi}||_2\leq\epsilon$, for some $\{\phi_i\}_{i=1}^{m'}.$

\State Perform Bell sampling on the copies of $\ket{\psi}.$

\State Run stabilizer learning algorithm to precision $\eta$ on the samples.

\State After stabilizer learning $\exists U_{C_1}^{\dagger}$ such that $ U_{C_1}^\dagger\ket\psi=\ket{\tilde\psi}\otimes\ket{0}^{m_1}$. Construct a basis of $\text{Ker}_{2,N}(\tilde\psi).$

\State Run the SBD algorithm on $\text{Ker}_{2,N}(\tilde\psi).$
\State Let $C_2$ be the block-diagonalising matrix output above, and $c'_i$ with $i\in[m_2]$ be the blocks.
\State Compute $\text{Ker}_{2,N}(U_{C_2}^\dagger\ket{\tilde\psi})$.

\For{$i\in[m_2]$}
    \State $P_i:=\Omega_+^{c'_i}$
\EndFor
  \State Solve the linear system $\sum_{i=1}^{m_2}s_iP_i\in\text{Ker}_{2,N}(U_{C_2}^\dagger\ket{\tilde\psi})$, with unknowns $\mathbf s=(s_1,\dots,s_{m_2})\in\mathbb F_2^{m_2}$.

\State \textbf{return} $U_{\tilde C}=U_{C_1}\cdot U_{C_2}$ \textit{and} $\{\tilde c_i\}_{i\in[m']}$ for $m'\leq m_1+m_2$, such that the corresponding subsystems are found pure.

\end{algorithmic}
\end{algorithm}

\subsection{Discussion and future work}

\paragraph{Hidden structure beyond conventional low-resource regimes.} Our results identify a family of efficiently learnable quantum states whose structure is not captured by any of the usual notions of low complexity. Namely, the states in \Cref{def:learning_task} can exhibit substantial entanglement, magic, and non-Gaussianity. Nevertheless, the hidden product structure leaves strong signatures in its Bell distribution, 
which can be retrieved efficiently.

Our approach to do so admits a crisp interpretation within the abelian StateHSP framework of Ref.~\cite{hinsche2026abelian}, but 
requires additional reconstruction beyond the ordinary hidden-cut problem. Accommodating the unknown Clifford enlarges the ambient group to account for all genuinely quadratic forms, and obscures the relation between the recovered symmetries and the underlying tensor-product structure. Our simultaneous-block diagonalisation step exploits the structure of the quadratic constraints to efficiently recover a frame in 
which the subsystem purity equations become explicit. 

This result can be, first of all, interpreted 
as an obstruction to a gluing lemma for Cliffords. 
It would be natural to think, given the high resource content that states in \Cref{def:learning_task} can have due to the mismatch between Cliffords and product unitaries, that a state like $U_C\bigotimes_i\ket{\psi_i}$ would be hard to learn. Potentially, by choosing $\{\ket\psi_i\}_i$ to be pseudorandom and using $U_C$ as a scrambling unitary, a plausible hypothesis would be that the transformed state would also be globally pseudorandom. Our result clearly disproves this and shows not only that Cliffords can't erase all evidence of the underlying factorisation, but also that this information can be efficiently accessed.

In this sense our results are closely related to those in Ref.~\cite{grewal2025efficient}, which gives efficient algorithms for learning unitaries of the form
\begin{equation}
    \label{eqn:unitaries}
    V=U_C^tU_d,\quad W=U_dU_C^t,
\end{equation}
where $U_C^t$ is a $t$-doped Clifford, and $U_d$ a depth-$d$ circuit. Their algorithm already showed that these apparently highly resourceful unitaries remained structured, and likewise identify certain block-structures in the Pauli expansion of operators related to $V$ and $W.$ Their work admits controlled departures from the exact Clifford and product regimes, by allowing modest amounts of magic and entanglement, but assumes a stronger oracle model of access to the unitaries and their inverses. Our work restricts to the exact case (exact Clifford and product) but works with a significantly weaker access model, that is only state access. This then also allows us to do unitary learning without demanding oracle access to the unitary or its inverse, but just its Choi state. Moreover, while their algorithm is improper, our method can sometimes manage to output a description of $V$ in terms of $U_C$ and $\bigotimes_iU_i$, as we show for the $t$-doped circuit case.

\paragraph{Consequences for $T$-doped Clifford circuits.} \Cref{thm:independent_pauli_learning} and \Cref{cor:t_doped_probability} deeply challenge previously established beliefs on the learnability of this class of circuits \cite{lai2022learning,leone2024learningstates,leone2024learning}. Indeed, general efficient learning algorithms were only known for $t$-doped circuits with $t=O(\log n)$, via compression methods. While these algorithms work at $t\sim n$, the runtime grows exponentially, deeming this regime quasi-chaotic \cite{leone2024learning}. The $t\sim 2n$ limit is then commonly associated with the onset of chaotic behaviour and no generic learning algorithm is believed to exist.

While our results do not contradict worst-case hardness in these regimes, they show that a large family of circuits of this type retains an efficiently recoverable compressed description well beyond the limit of $t=O(\log n)$, in terms of states of the form in \Cref{def:learning_task}. Moreover, far from being an isolated effect, an exponentially big fraction of these circuits satisfy this in the \textit{quasi-chaotic} regime ($t\sim n$), and still a constant fraction in the \textit{chaotic} case ($t\sim 2n$), and hence can be learnt efficiently. Besides these cases, our algorithm can still learn certain circuits in the setting where $t>2n$. For instance, whenever the algebraic dependencies breaking the underlying product structure would remain confined in subsystems of $O(\log n)$ many qubits, the unitaries could still be learnt. It could then be relevant to study, for instance, what circuits in this class satisfy this property.

In this regard our work also closely 
connects to Ref.~\cite{lee2025learning}, where efficient learning algorithms are given for states with higher doping than $O(\log(n))$, albeit for the restricted case where these can be written as $\ket{\psi}=T^{\otimes\vec i}U_C\ket{0}^{\otimes n}$. Certain members of this state family are also accessible to our framework after conjugating the $T$ gates through. Moreover, the corresponding unitary family $U=T^{\otimes\vec i}U_C$ is subsumed by our \Cref{thm:independent_pauli_learning}. Their approach however, differs significantly from ours by looking at the problem as a non-abelian 
StateHSP instance. We believe that exploring these connections could be interesting to understand the full reach of efficient learnability of $t$-doped circuits.

\paragraph{Benchmarking resource-rich regimes and logical sectors.} We believe that our results may also be relevant in practical scenarios, by enabling probes of previously inaccessible regimes. Many benchmarking and certification methods need to exploit some efficiently accessible description of the states and circuits to remain scalable. To date, these have then been restricted, mostly to low-resource regimes, relying on low-bond-dimension MPS, almost stabilizer structure, or Gaussianity of states \cite{eisert2020quantum,aolita2015reliable,cramer2010efficient,flammia2011direct,kliesch2021theory}. Nevertheless, the states and unitaries that we consider here can still look highly complex according to any such description, and hence our methods could provide access to a regime complementary to that of previous works. This could be especially relevant for designing new benchmarking protocols for instance for the early fault-tolerant regime \cite{mills2025logical,harper2019fault,leitch2026learningunknownstabilizercodes}, where extending the reach of previous techniques will likely be required.

A natural task for our algorithm in that setting would for instance be to perform learning of logical product states of stabilizer error correcting codes \cite{gottesman1997stabilizer}. Indeed, if $U_C$ is the Clifford encoding circuit of a stabilizer code, and the logical state factors across logical subsystems, then its physical realization is of the form
\begin{equation}
    \ket{\psi}=U_C\left(\bigotimes_i\ket{\phi_i}_{L}\otimes\ket{0}_{\text{aux}}\right),
\end{equation}
which can be learnt by \Cref{alg:hidden-clifford-test}. Recovering the hidden factors, could therefore be used to certify the absence of spurious correlations between distinct logical subsystems, in a completely code agnostic way, even when the systems are not physically separated.

\paragraph{Hamiltonian property learning.}

Other potential applications of our algorithm include the study of many-body systems, where one might be interested in learning some Hamiltonian $H$, or at least some of its properties. A broad range of methods has been developed for this purpose \cite{wiebe2014hamiltonian,bairey2019learning,anshu2021sample,huang2023learning,haah2024learning, hu2025ansatz,bluhm2026certifying, bluhm2026hamiltonian, ma2024learning, franca2024efficient}, but their guarantees typically rely on short evolution times, or strong structural facts such as locality or sparsity in some other specified operator basis. For Hamiltonians of the type
\begin{equation}
    H=U_C\left(\sum_S H_S\otimes \mathbb I_{S^c}\right)U_C^\dagger,
\end{equation}
where common promises might not hold, our \Cref{alg:hidden-clifford-test} can still manage to learn $H$ given access only to its forward time evolution operator, at almost arbitrary times. While the algorithms in Ref.~\cite{grewal2025efficient} could also be used for this purpose, the fact that our algorithm does not rely on inverse time access makes it significantly more amenable to be run in practice. 

Moreover, even when full Hamiltonian learning might still not be possible, our methods can still help us finding decoupling strategies that drastically compress the system and reduce its simulation cost for instance. These phenomena have been studied recently in the context of critical spin chains \cite{fan2025disentangling,kim2026disentangling}, where it has been shown that $XX$-type Hamiltonians can decouple exactly into two distinct Ising chains under the action of a Clifford, in approaches building on earlier related work disentangling fermionic systems
\cite{PhysRevLett.117.210402}.
Our algorithm could then enable practical exploration of these hidden symmetries for more families of systems.

\paragraph{Open questions.}
Whether or not our guarantees can be extended to a robust setting is one of the main open questions we leave for the future. As we already stated, in Ref.~\cite{bakshi2025learning}, an efficient learning algorithm is given also for the case where the Cliffords showcase some amount of magic and the product structure is slightly perturbed to have some entanglement. This is only possible in their case given full oracle access to the unitary and its inverse. In our case, under these perturbations, namely for state vectors 
of the form
\begin{equation}
    U_C^tU_d\ket{0}^{\otimes n},
\end{equation}
with $U_C^t$ a $t$-doped Clifford, and $U_d$ an arbitrary depth-$d$ circuit, for $t,d$ small enough, the exact quadratic constraints become noisy. Recovering them by our method then leads to a parity-learning problem resembling learning parity with noise (LPN) \cite{blum1994lpn}. This identifies a clear computational obstruction for our algorithm to work in that case, albeit some more involved protocols could potentially overcome this. We conjecture however, that any such method necessarily runs into similar computational issues, rendering these states hard to learn.

States of this form live in the first level of the magic hierarchy introduced in Ref.~\cite{parham2026quantum}, and understanding their learnability would have strong implications for questions across a wide range of fields, like cryptography, or resource theories. In the latter context, these states have recently received significant amount of attention under the name of \emph{Clifford augmented MPS} (CAMPS) \cite{qian2024camps,masot2026limits,fux2025disentangling,liu2026camps,fan2025disentangling,kim2026disentangling}, as they have been argued to be a better ansatz for bounding the simulation complexity of some many-body systems, and thus solving this conjecture would shed some light on the interplay between simulability and learnability of quantum states.

A further open question is whether the same strategy admits a genuinely fermionic formulation. In particular, can analogous quadratic symmetries in the Majorana representation be efficiently sampled and simultaneously block-diagonalised so as to recover a hidden tensor-product or mode decomposition up to a fermionic Gaussian transformation? Such an extension could provide learning and disentangling algorithms directly in the fermionic setting, without first mapping the system to qubits, while naturally incorporating fermionic parity and the absence of an ordinary subsystem tensor-product structure of fermionic modes. Also in that case, extensions to related free fermionic augmented MPS \cite{PhysRevLett.117.210402} could be interesting.

Finally, on a more technical note, our analysis gives an explicit transformation between quadratic symmetries of the Bell distribution and quadratic symmetries of the Pauli distribution, showing that its image is a proper subset of the latter. In particular, some properties of the Pauli distribution cannot be straightforwardly inferred from quadratic symmetries of the Bell samples. A better characterization of this, would sharpen our understanding of which nonlinear properties of pure states are accessible through Bell measurements. We conjecture that all quadratic Bell constraints necessarily reduce to stabilizer and subsystem purities (up to affine transformations). Proving this, would both simplify the analysis developed here, and clarify the intrinsic information content of the Bell distribution. Conversely, finding genuinely different families of symmetries, would then point to new classes of structured states admitting similar efficient learning algorithms.

\subsection*{Structure of 
this work}

We introduce some preliminaries about Bell sampling in \Cref{sec:bell_sampling}. We introduce our framework of quadratic Bell symmetries and generalized purities and show how some canonical examples fit into it. In \Cref{sec:clifford_dressed} we then introduce the family of states that we care about and show how they are still structured in a very particular way. Given this structure we already show here that there exists an exponential runtime algorithm requiring only access to polynomially many copies in order to retrieve it. In \Cref{sec:comp_eff} we lift this to an algorithm running in polynomial time. We introduce the required finite algebra background and show the simultaneous block-diagonalisation procedure. We complete the section showing that the sample complexity of this algorithm is still polynomial. We finalize in \Cref{sec:Tdoped} by showing how to disentangle the Choi states of $T$-doped unitaries, and how to use this to come up with a proper learning algorithm for them.

\section{Bell sampling}

\label{sec:bell_sampling}

We now introduce Bell sampling and its associated probability distribution, which form the foundation of our entire framework. Bell sampling has emerged as a remarkably powerful primitive in quantum learning theory: collective measurements on pairs of copies can reveal algebraic structure that is inaccessible to ordinary single-copy measurements, enabling, among other applications, efficient stabilizer learning and testing, fidelity estimation, and the recovery of hidden tensor-product structure \cite{montanaro2017learning,gross2021schur,grewal2024improved,grewal2025efficient,hangleiter2024bell,hinsche2025efficient,hinsche2026abelian}. In our setting, its central role is that the Bell distribution retains precise quadratic signatures of purity and product structure even after an unknown Clifford transformation has obscured the underlying tensor-product decomposition.


\begin{definition}[{[Bell distribution $\mathsf b_\rho$]}]
    \label{def:bell_dsitribution}
    Let $\rho$ be some state on $n$ qubits. Then let $d=2^n,$ and
    \begin{equation}
        \ket{\Phi}=\frac{1}{\sqrt d}\sum_{x\in\{0,1\}^n}\ket{x}\ket{x}
    \end{equation}
    be the maximally entangled state on two copies of an $n$-qubit system. For all $n-$qubit Pauli operators 
    \begin{equation*}
    \ket{\Phi_P}=(P\otimes \mathbb I)\ket{\Phi},
    \qquad P\in\mathcal P_n .
\end{equation*}
define the full basis of Bell states on $2n$-qubits. 

Then, given two copies of the state $\rho$, Bell sampling performs measurements of $\rho\otimes\rho$ in the Bell basis. The Bell distribution of $\rho$ is then defined as the probability distribution resulting from these measurements
\begin{equation}
\mathsf b_\rho(\Phi_P)=\Tr\!\left[\ket{\Phi_P}\!\bra{\Phi_P}\,\rho\otimes\rho\right]=\frac{1}{d}\Tr\!\left[P\rho P\rho^T\right].
\end{equation}
\end{definition}
Since all Bell states can be labeled by their corresponding Pauli, slightly abusing notation, from now on we will label the Bell distribution using Paulis directly. In spite of this, one should take care to not confuse this with the Pauli distribution of a state, i.e.,
\begin{equation}
    \mathsf P_\rho(P)=\frac{1}{d}\text{Tr}[P\rho]^2.
\end{equation}
This difference becomes particularly clear for pure states $\rho=\ket\psi\!\!\bra\psi$, for which the Bell distribution reduces to 
\begin{equation}
    \label{eqn:pure_Bell}
    \mathsf b_\psi(P)=\frac{1}{d}|\langle\psi|P|\psi^*\rangle|^2.
\end{equation}
From this, one can see that, while the two distributions coincide in the special case of real states, they do differ in general.

For our purposes, the important point about the Bell distribution is that it will help us convert the structural properties of our states into a clear algebraic constraint on its support. Indeed, certain properties of $\rho$ will set certain probabilities $\mathsf b_\rho(P)$ to zero in some structured way. For instance, we will show that stabilizer and product structure, restrict the Bell distribution to live in some quadric defined by some binary matrix $M\in M_{2n}(\mathbb F_2)$ and the equation
\begin{equation}
     p^TMp=c,\quad\forall p\in\supp(\mathsf b_\rho)
\end{equation}
with $p\in\mathbb F_2^{2n}$ the symplectic representation of the Bell labels. We will prove this in the following section and develop it further to see how this helps in solving our problem.

\subsection{Quadratic symmetries of the Bell distribution}

For the moment, we will ignore sampling-complexity considerations and assume that we have access to the full support of the Bell distribution
\begin{equation*}
    \supp(\mathsf b_\rho)=\{q\in\mathbb F_2^{2n}:\mathsf b_\rho(P_q)\neq 0\}.
\end{equation*}
Our goal is to show how several structural properties of a state, such as purity, stabilizers, or product structure across hidden cuts, manifest as quadratic constraints on this support. To that end, let us consider the maximally mixed state $\rho=\mathbb I/d$, for which the Bell distribution looks completely unstructured:
\begin{equation}
    \mathsf b_{\mathbb I/d}(P)=\frac{1}{d}\Tr\!\left[P\frac{\mathbb I}{d}P\frac{\mathbb I}{d}\right]=\frac{1}{d^2}\qquad\forall P\in\mathcal P_n .
\end{equation}
Thus the Bell distribution is uniform over all $d^2$ Pauli labels, and there are no non-trivial support constraints.

For pure states the situation changes already. In that case, it can be shown that Bell outcomes belonging to the antisymmetric subspaces have zero probability, and thus one can't ever sample an odd number of singlets
\begin{equation}
    \ket{\Psi_{\text{singlet}}}=\frac{\ket{0,1}-
    \ket{1,0}}{\sqrt{2}}=\ket{\Phi_Y}.
\end{equation}
Since singlets are labeled by the antisymmetric single qubit $Y$ Pauli, in the $n$-qubit setting, this constraint means that $\mathsf b_\psi$ must be entirely supported on Bell states labeled by real Paulis. Indeed, one can show that for Paulis such that
\begin{equation}
    P^T=-P
\end{equation}
it must be that 
\begin{equation}
    \mathsf b_\psi(P)=0.
\end{equation}
This can be seen from looking at \Cref{eqn:pure_Bell}. Writing
\begin{equation}
    s=\bra{\psi}P\ket{\psi^*},
\end{equation}
since $s$ is a scalar, $s=s^T$ and therefore
\begin{align}
    s=s^T&=((\ket{\psi^*}^T)P\ket{\psi^*})^T\\&=((\ket{\psi^*}^T)P^T\ket{\psi^*})\notag\\&=\bra{\psi}P^T\ket{\psi^*}.\notag
\end{align}
If $P^T=-P$, then this implies $s=-s$ and hence $s=0$. Therefore
\begin{equation}
    \mathsf b_\psi(P)=0
    \qquad
    \text{whenever } P^T=-P .
\end{equation}

This constraint has a particularly transparent form as a quadratic symmetry in symplectic representation. In fact, write
\begin{equation}
    q=(x_1,z_1,\dots,x_n,z_n)\in\mathbb F_2^{2n}
\end{equation}
such that
\begin{equation}
    P_q=\prod_{i=1}^nX_i^{x_i}Z_i^{z_i}.
\end{equation}
Then a single-qubit factor is antisymmetric when both $x_i=z_i=1$. Thus the global parity of imaginary terms in the $n-$qubit case reads
\begin{equation}
    \sum_{i=1}^nx_iz_i=q^T\Omega_+q,\qquad \Omega_+=\bigoplus_{i=1}^n\begin{pmatrix}
        0 &1\\ 0& 0
    \end{pmatrix}.
\end{equation}
Hence the purity constraint can be written as the quadratic equation
\begin{equation}
    \label{eqn:purity_quadric}
    q^T\Omega_+q=0,\qquad \forall q\in\supp(\mathsf b_\psi).
\end{equation}
We can now make this discussion slightly more general in order to see how, for instance, product structure gives a direct refinement of the global purity case $\Omega_+$. We will also show that in fact stabilizers then arise as linear constraints of the same type.
For this, let $\rho$ be a generic $n-$qubit quantum state. Then, we can take any $M\in\text{Mat}_{2n}(\mathbb F_2)$ and define a quadratic symmetry on the Bell distribution such that
\begin{equation}
     q^TMq=c,\quad \forall q\in \supp(\mathsf b_\rho),
\end{equation}
where $c\in\mathbb F_2$ is constant. In order to help interpret these constraints we will show a direct translation between these quadratic Bell symmetries and symmetries in the Pauli distribution. Before doing so however, let us remind some basic facts of binary quadratic forms that will be useful for that.

Firstly, since we work over $\mathbb F_2$, every $M$ is totally determined by its associated bilinear form $B=M+M^T$ and its diagonal $v=\text{diag}(M),$ since  
\begin{equation}
    \label{eqn:quadratic_form}
    q^TMq=\sum_{i=1}^{2n}M_{i, i}q_i+\sum_{i<j}(M_{i, k}+M_{j , i})q_iq_j=\sum_{i=1}^{2n}M_{i, i}q_i+\sum_{i<j}B_{i, k}q_iq_j.
\end{equation}
 Thus, we can assume $M$ to be an upper triangular matrix, since for any symmetric term $S$, $M=M+S$ as quadratic forms. Moreover, the polar form $B$ is symmetric and alternating, and hence by standard results in classification of bilinear forms \cite{grove2002classical}, there exist some $A\in GL(2n,\mathbb F_2)$, such that
\begin{equation}
    A^{-T}BA^{-1}=\Omega^l\oplus 0_{2n-2l},\quad \Omega^l=\bigoplus_{i=1}^l\begin{pmatrix}
        0&1\\1&0
    \end{pmatrix}_i,
\end{equation}
and by \Cref{eqn:quadratic_form}, $A$ also block-diagonalises $M$ as a quadratic form, since for $M'=A^{-T}MA^{-1}$
\begin{equation}
    q^TM'q=\sum_{i=1}^{2n}M'_{i, i}q_i+\sum_{i<j}B'_{i, k}q_iq_j=q^T(v'+\Omega_+^l)q
\end{equation}
for $v'=\text{diag}(M')$, which is clearly block-diagonal. Hence, for all practical purposes
\begin{equation}
    \label{eqn:block_quadratic}
    M=A^{T}(\Omega^l_++v')A
\end{equation}
 We can now state our Lemma.

\begin{lemma}[{[Quadratic Bell symmetries induce Pauli constraints]}]
    \label{lemma:bell_to_Pauli}
    Let $\rho$ be an $n$-qubit state satisfying a quadratic Bell symmetry $M$. Then its Pauli distribution satisfies the following generalised purity equation
    \begin{equation}
        1=\frac{1}{2^l}\sum_{\substack{P\in\mathcal P_n:\\ \Omega p\in A^T{v'}+\mathrm{im}(M+M^T)}}\text{Tr}[P\rho]^2\cdot(-1)^{p^TN_Mp+c_{v'}+c}
    \end{equation} 
    for some quadratic symmetry $N_M$ given by 
    \begin{equation}
    \label{eqn:Pauli_symmetry}N_M=\Omega_++\Omega^TA^{-1}\big[\Omega_+^l+\Omega^lv'v'^T\Omega^l\big]A^{-T}\Omega
    \end{equation}
    and where $l=\frac{1}{2}\text{rank}(M+M^T)$, $A\in GL(2n,\mathbb F_2)$ is the block-diagonalising transformation, and $c_{v'}$ is an additional shift that depends on $v'$ coming from \Cref{eqn:block_quadratic}.
\end{lemma}

\begin{proof}
    
    First, the condition imposed by $M$ can be written as a generalized purity constraint on the Bell distribution. Namely,
    \begin{equation}
    1=\frac{1}{d}\sum_{\substack{Q\in\mathcal P_n:\\q\in  \supp(\mathsf b_\rho)}}\text{Tr}[Q\rho Q\rho^T]=\frac{1}{d}\sum_{Q\in\mathcal P_n}\text{Tr}[Q\rho Q\rho^T]\cdot(-1)^{q^TMq+c}
    \end{equation}
where in the last equality we have used that $\forall Q$ s.t. $q^TMq+c=1$, then $\text{Tr}[Q\rho Q\rho^T]=0$. 
Now, we can proceed by writing $\rho$ in the Pauli basis
\begin{equation}
    \rho=\frac{1}{d}\sum_{P\in\mathcal P_n} \text{Tr}[P\rho]P.
\end{equation}
Plugging this in the previous expression yields the following
\begin{align}
    1&=\frac{1}{d}\sum_{Q\in\mathcal P_n}\text{Tr}[Q\rho Q\rho^T]\cdot(-1)^{q^TMq+c}\\ 
        &=\frac{1}{d^2}\sum_{P,Q\in\mathcal P_n}\text{Tr}[P\rho]^2\text{Tr}[QPQP^T]\cdot(-1)^{q^TMq+c}\notag\\
        &=\frac{1}{d^2}\sum_{P,Q\in\mathcal P_n}\text{Tr}[P\rho]^2\cdot(-1)^{p^T\Omega_+p+q^T\Omega p}\cdot(-1)^{q^TMq+c}\notag\\
        &=\frac{1}{d^2}\sum_{P\in\mathcal P_n}\text{Tr}[P\rho]^2\cdot(-1)^{p^T\Omega_+p}\sum_{Q\in\mathcal P_n}(-1)^{q^T\Omega p+q^TMq+c},\notag
\end{align}
where following the previous notation $\Omega=\Omega^n$ is the standard symplectic form. Now we can analyze the sum over $Q$, for which is useful to write $M$ as in \Cref{eqn:block_quadratic}
\begin{align}
    \label{eqn:symmetry_translation}
    \sum_{Q\in\mathcal P_n}(-1)^{q^T\Omega p+q^TMq+c}&=\sum_{Q\in\mathcal P_n}(-1)^{q^T\Omega p+q^TA^T(\Omega_+^l+v')Aq+c}\\
        &=\sum_{Q\in\mathcal P_n}(-1)^{(A^{-1}q)^T\Omega p+q^T(\Omega_+^l+v')q+c}\notag\\
        &=\sum_{Q\in\mathcal P_n}(-1)^{q^T(v'+A^{-T}\Omega p)+q^T\Omega_+^lq+c}\notag\\
        &=\sum_{Q\in\mathcal P_n}(-1)^{q^Tv'_p+q^T\Omega_+^lq+c}\notag\\
        &=(-1)^c\prod_{i=1}^n\sum_{Q\in\mathcal P_1}(-1)^{q^T{v'}_p^i+q^T\Omega_+^{i_l}q}\notag
\end{align}
where in the second equality we have applied a change of basis on the Paulis, i.e., $q\leftarrow Aq$, and in the third and fourth equalities we have used that over $\mathbb F_2$ we have $a^2=a$, and thus
\begin{equation}
    q^Tv'q=\sum_{i=1}^{2n}v'_{i, i}q_i^2=\sum_{i=1}^{2n}v'_{i, i}q_i=q^T\cdot \vec v'
\end{equation}
and we can treat $v'$ as a vector (we will however write $v'$ instead of $\vec v'$ for ease of notation). Finally, in the last equality we take advantage of the block structure and factorise the sum into single qubit sums, and where $i_l=1$ if $i\leq l$ and $i_l=0$ otherwise. Then evaluating the single qubit sums yields the following result
\begin{equation}
    \label{eqn:sum_over_Q}
    \sum_{Q\in\mathcal P_n}(-1)^{q^T\Omega p+q^TMq+c}=(-1)^c\prod_{i=1}^n\sum_{Q\in\mathcal P_1}(-1)^{q^T{v'}_p^i+q^T\Omega_+^{i_l}q}=\begin{cases}
        0 & {v'}_p\not\in \mathrm{im}(\Omega^l),\\
        4^{n-l}\times 2^l(-1)^{v'_p\Omega^l_+v'_p+c} &\mathrm{else}
        .
    \end{cases}
\end{equation}
Indeed, let us evaluate one single sum term, to get
\begin{equation}
    \sum_{Q\in\mathcal P_1}(-1)^{q^T{v'}_p^{i}+q^T\Omega_+^{i_l}q}=\sum_{q_1,q_2\in\mathbb F_2}(-1)^{q_1{v'}_p^{i,1}+q_2{v'}_p^{i,2}+q_1q_2\cdot\delta_{i\leq l}}
\end{equation}
For $i>l$ this becomes
\begin{equation}
    \sum_{q_1,q_2\in\mathbb F_2}(-1)^{q_1{v'}_p^{i,1}+q_2{v'}_p^{i,2}}=\bigg(\sum_{q_1\in\mathbb F_2}(-1)^{q_1{v'}_p^{i,1}}\bigg)\cdot\bigg(\sum_{q_2\in\mathbb F_2}(-1)^{q_2{v'}_p^{i,2}}\bigg)
    =\begin{cases}
        4,& ({v'}_p^{i,1},{v'}_p^{i,2})=(0,0)\\
        0,& ({v'}_p^{i,1},{v'}_p^{i,2})\neq(0,0)
    \end{cases}.
\end{equation}
Thus, the full product vanishes unless ${v'}_p^i=0$ $\forall i>l,$ which is equivalent to saying $v'_p\in\text{im}(\Omega_l)$, since 
\begin{equation}
    \text{im}(\Omega_l)=\{v'\in\mathbb F_2^{2n}:{v'}^i=0\ \forall i>l\}.
\end{equation}
On the other hand, for $i\leq l$
\begin{equation}
    \sum_{q_1,q_2\in\mathbb F_2}(-1)^{q_1{v'}_p^{i,1}+q_2{v'}_p^{i,2}+q_1q_2}=\sum_{q_1\in\mathbb F_2}(-1)^{q_1{v'}_p^{i,1}}\sum_{q_2\in\mathbb F_2}(-1)^{({v'}_p^{i,2}+q_1)q_2}=2(-1)^{({v'}_p^i)^T\Omega_+^i{v'}_p^i}
\end{equation}
since in the second equality, the second sum is nonzero only when $q_1={v'}_p^{i,2}.$ This then yields the above result. Now, we can plug \Cref{eqn:sum_over_Q} into \Cref{eqn:symmetry_translation}, to get
\begin{align}
    1 &= \frac{1}{2^l}\sum_{P\in\mathcal P_n:v'_p\in\mathrm{im}(\Omega^l)}\text{Tr}[P\rho]^2\cdot(-1)^{p^T\Omega_+p+{v'_p}^T\Omega_+^lv'_p+c}\\
    &= \frac{1}{2^l}\sum_{P\in\mathcal P_n:v'_p\in\mathrm{im}(\Omega^l)}\text{Tr}[P\rho]^2\cdot(-1)^{p^T\Omega_+p+(v'+A^{-T}\Omega p)^T\Omega_+^l(v'+A^{-T}\Omega p)+c}\notag\\
    &= \frac{1}{2^l}\sum_{P\in\mathcal P_n:v'_p\in\mathrm{im}(\Omega^l)}\text{Tr}[P\rho]^2\cdot(-1)^{p^T\Omega_+p+p^T\Omega(A^{-1}\Omega_+^lA^{-T})\Omega p+{v'}^T\Omega^l A^{-T}\Omega p+c_{v'}+c}\notag\\
    &= \frac{1}{2^l}\sum_{P\in\mathcal P_n:v'_p\in\mathrm{im}(\Omega^l)}\text{Tr}[P\rho]^2\cdot(-1)^{p^T\Omega_+p+p^T\Omega A^{-1}(\Omega_+^l+\Omega^l v'{v'}^T\Omega^l) A^{-T}\Omega p+c_{v'}+c}\notag\\
    &=\frac{1}{2^l}\sum_{\substack{P\in\mathcal P_n:\\ \Omega p\in A^Tv'+\mathrm{im}(M+M^T)}}\text{Tr}[P\rho]^2\cdot(-1)^{p^TN_Mp+c_{v'}+c}\notag
\end{align}
where in the third equality we have used that 
\begin{equation}
    {v'}^T\Omega_+^lA^{-T}\Omega p+p^T\Omega A^{-1}(\Omega_+^l)v'={v'}^T(\Omega_+^l+(\Omega_+^l)^T)A^{-T}\Omega p={v'}^T\Omega^lA^{-T}\Omega p,
\end{equation}
and where
\begin{equation}
    c_{v'} = {v'}^T \Omega_+^l v'
\end{equation}
and
\begin{equation}
    N_M=\Omega_++\Omega^TA^{-1}\big[\Omega_+^l+\Omega^lv'v'^T\Omega^l\big]A^{-T}\Omega.
\end{equation}
Thus, every quadratic symmetry $M$ at the level of the Bell distribution leads to a quadratic symmetry $N_M$ of the form above on the Pauli distribution of the state and we have proven our statement. 
\end{proof}

We believe that \Cref{lemma:bell_to_Pauli} is interesting on its own. Indeed, not all quadratic forms can be written as in \cref{eqn:Pauli_symmetry}, and hence it helps draw a criterion for when certain symmetries of the Pauli distribution of a state can be observed from quadratic symmetries of the Bell distribution or not. An example of this is the case of imaginarity in the Pauli basis, which counts the weight of odd-$Y$ Paulis. The quadratic form 
\begin{equation}
    N_{\text{im}}= \Omega_+.
\end{equation}
precisely records this quantity. Then, a state $\rho$ is real if and only if its Pauli support satisfies $p^TN_{\text{im}}p=0$. Despite this simple quadratic description, reminiscent of the purity constraint, efficient measurement of such a property is in general not possible in the copy-access model \cite{schatzki2024imaginarity,hinsche2025efficient}. Interestingly, an algebraic impediment can already be seen from our \cref{eqn:Pauli_symmetry}. Indeed, by plugging $N_{im}$ in it, we obtain
%
\begin{equation}
    0=\Omega^TA^{-1}[\Omega^l_++\Omega^lvv^T\Omega^l]A^{-T}\Omega\iff l=0.
\end{equation}
Thus, the corresponding Bell constraint has zero polar form and is therefore linear. However, as we show below, that can at most correspond to a stabilizer symmetry and cannot determine the total odd-$Y$ weight for systems larger than one qubit. Hence, our correspondence prevents the existence of a genuinely quadratic Bell-support constraint that measures this property, in agreement with the fact that it cannot
be efficiently measured.

On the other hand, any binary matrix $M\in\mathbb F_2^{2n\times 2n}$ can define a candidate Bell symmetry, and \Cref{lemma:bell_to_Pauli} can be useful to interpret their meaning. For instance we can now use it to immediately see that the former case, where $M=\Omega_+$ and $c=0$, indeed corresponds to a global purity condition. This follows since in that case $A=\mathbb  I_{2n}$ and $v=0$
\begin{equation}
    N_{\Omega_+}=\Omega_++\Omega^T\Omega_+\Omega=\Omega
\end{equation}
and since $\mathrm{im} (M+M^T)=\mathrm{im}(\Omega)=\mathcal P_n$, we recover the standard purity equation
\begin{equation}
    1=\frac{1}{2^n}\sum_{P\in \mathcal P_n}\text{Tr}[P\psi]^2(-1)^{p^T\Omega p}=\frac{1}{2^n}\sum_{P\in \mathcal P_n}\text{Tr}[P\psi]^2.
\end{equation}

Similarly, product structure is then simply a refinement of the global purity constraint to smaller subsystems. Indeed, given a product state vector $\ket{\psi}=\ket{\psi_A}\otimes\ket{\psi_B}$, since subsystems $A$ and $B$ are also pure, $\ket{\psi}^{\otimes 2}$ is not only invariant under global SWAPs, but also local SWAPs of the $A$ and $B$ subsystems. In other words, at the level of the Bell distribution, both symmetries 
\begin{equation}
    \Omega_+^A=\Omega_+|_A\oplus 0_B,\qquad \Omega_+^B=0_A\oplus\Omega_+|_B
\end{equation}
again with $c=0$, must be satisfied. Indeed, we can observe that these restrictions lead to the standard subsystem purity equations in the Pauli distribution. Again, $A=\mathbb I_{2n}$ and $v=0$, but now $l_i=n_i$, and $\mathrm{im}(\Omega_+^i+{\Omega_+^{i}}^T)=\mathcal P_{n_i}$. Then, for instance for $i=A$, we get
\begin{align}
    \label{eqn:subsystem_purity_symm}
    1&=\frac{1}{d_{A}}\sum_{P_A\in\mathcal P_{n_A}}\text{Tr}[(P_A\otimes\mathbb I)\psi]^{2}(-1)^{p_A^T(\Omega_++\Omega_-^A)p_A}\\
    &=\frac{1}{d_A}\sum_{P_A\in\mathcal P_{n_A}}\text{Tr}[(P_A\otimes\mathbb I)\psi]^2(-1)^{p_A^T(\Omega^A\oplus\Omega_+^B)p_A}\notag\\
    &=\frac{1}{d_A}\sum_{P_A\in\mathcal P_{n_A}}\text{Tr}[P_A\psi]^2=\text{Tr}[\psi_A^2].\notag
\end{align}

Where $\Omega_-^A={\Omega_+^{A}}^T$, and we slightly abused notation for $p_A=p_A\oplus 0_B.$
Finally, it can also be seen that stabilizer symmetries correspond to linear symmetries on the Bell support, and in fact these are homogeneous (i.e., $c=0$) for real stabilizers, and affine ($c=1$) for imaginary stabilizers. In that case, $A=\mathbb I_{2n}$, $l=0$, and $\text{span}(M+M^T)=0$. This results in a collapse of the purity equation to a single term, since the condition $\Omega p\in A^Tv+\mathrm{im}(M+M^T)$ becomes 
\begin{equation}
    \label{eqn:stabilizer_symm_1}
    p_*=\Omega A^Tv=\Omega v.
\end{equation}
On the other hand, since $l=0$ the phase term $N_M$ becomes
\begin{equation}
    \label{eqn:stabilizer_symm_2}
    N_v=\Omega_+.
\end{equation}
Thus, overall 
\begin{equation}
    \label{eqn:stabilizer_symm_3}
    1=\text{Tr}[P_*\psi]^2(-1)^{p_*^T\Omega_+p_*+c}.
\end{equation}
For this to hold, $c$ must satisfy
\begin{equation}
    \label{eqn:stabilizer_symm_4}
    c=p_*^T\Omega_+p_*
\end{equation}
i.e., $c=0$ if $P_*$ is real, and $c=1$ otherwise, and then
\begin{equation}
    \label{eqn:stabilizer_sym_5}
    1=\text{Tr}[P_*\psi]^2,
\end{equation}
which is a standard stabilizer condition. Moreover, in the following Lemma, we will show that some symmetries of the form $M=\Omega_+^S+d$, i.e., subsystem purity symmetries plus some diagonal component, can actually be decoupled into the corresponding purity and stabilizer symmetries, which will be a useful result for later.

\begin{lemma}[{[Diagonal symmetries imply product and stabilizer structure]}]
    \label{lemma:diagonal_stab}
    Given a state vector $\ket{\psi}$ such that $M=\Omega_+^S+d\in\text{Ker}_2(\psi)$, for some subsystem $S$ and $d$ some diagonal term, then the state must be a product and if $d\neq 0$ then $\exists  P\in\text{STAB}(\psi)$.
\end{lemma}
\begin{proof}
    In order to see this, we will look at the Pauli distribution constraint instead, since it is easier to interpret. For a symmetry $\Omega_+^S+d$, we take $M=\Omega_+^S$, $v=d$ and $A=\mathbb I_{2n}$.

    In this setting
    \begin{align}
        N_M&=\Omega_++\Omega\Omega_+^S\Omega+\Omega\Omega^Sdd^T\Omega^S\Omega\\
        &=\Omega^S+\Omega_+^{S^c}+d_S(d_S)^T\notag .
    \end{align}
    On the other hand, the sum runs over
    \begin{equation}
        p\in \Omega d+\mathrm{im}(\Omega\cdot\Omega^S)=\Omega d_S+\mathrm{im}(\mathbb I_S)+q,\quad q=\Omega d_{S^c},
    \end{equation}
    so that the component on $S^c$ is fixed.
    Plugging this into the Pauli purity equation yields
    \begin{align}
        1&=\frac{1}{2^{|S|}}\sum_{p_S\in\Omega d_S+\mathrm{im}(\Omega^S)}\text{Tr}[P_S\otimes Q\psi]^2\cdot(-1)^{p^T(\Omega^S+\Omega_+^{S^c}+d_Sd_S^T)p+d^T\Omega_+^Sd+c}\\
        &=\frac{1}{2^{|S|}}\sum_{p_S\in\Omega d_S+\mathrm{im}(\Omega^S)}\text{Tr}[P_S\otimes Q\psi]^2\cdot(-1)^{p^T\Omega^Sp+q^T\Omega_+^{S^c}q+p^Td_Sd_S^Tp+d^T\Omega_+^Sd+c}\notag\\
        &=\frac{1}{2^{|S|}}\sum_{p_S\in\Omega d_S+\mathrm{im}(\Omega^S)}\text{Tr}[P_S\otimes Q\psi]^2\cdot(-1)^{p^Td_Sd_S^Tp}\notag\\
        &=\frac{1}{2^{|S|}}\sum_{p_S\in\mathrm{im}(\Omega^S)}\text{Tr}[P_S\otimes Q\psi]^2\cdot(-1)^{(p+\Omega d_S)^Td_Sd_S^T(p+\Omega d_S)}\notag\\
        &=\frac{1}{2^{|S|}}\sum_{p_S\in\mathrm{im}(\Omega^S)}\text{Tr}[P_S\otimes Q\psi]^2\cdot(-1)^{p^Td_Sd_S^Tp+p^Td_Sd_S^T\Omega d_S+d_S^T\Omega d_Sd_S^Tp+d_S^T\Omega d_Sd_S^T\Omega d_S}\notag\\
        &=\frac{1}{2^{|S|}}\sum_{p_S\in\mathrm{im}(\Omega^S)}\text{Tr}[P_S\otimes Q\psi]^2\cdot(-1)^{p^Td_Sd_S^Tp}\notag\\
        &=\frac{1}{2^{|S|}}\sum_{p_S\in\mathrm{im}(\Omega^S)}\text{Tr}[P_S\otimes Q\psi]^2\cdot(-1)^{ d_S^T\cdot p}\notag
    \end{align}
    where we have stripped off all constants, since they are just there for consistency and ultimately cancel each other out.

 From this it is easier to interpret the effect of the diagonal component $d$, which we have now split into $d_S$ and $Q$, being the components of the diagonal acting on $S$ and $S^c$ respectively. In order to see that both $Q\neq\mathbb I$ and $d_S\neq 0$ impose stabilizer conditions we can do the following. First assume the case where $Q=\mathbb I$, then
\begin{equation}
    1=\frac{1}{2^{|S|}}\sum_{P\in\mathcal P_S}\text{Tr}[P\psi]^2\cdot(-1)^{d_S^T\cdot p}\leq\frac{1}{2^{|S|}}\sum_{P\in\mathfrak{\mathcal P_S}}\text{Tr}[P\psi]^2\leq1.
\end{equation}
Which is a contradiction except in the case of exact equality. For the first $\leq$ to be an exact equality we must have that $p^T\cdot d_S=0$ $\forall p\in\mathcal P_S$ with $\text{Tr}[P\psi]^2>0$, which is exactly a stabilizer condition, since this demands that all Paulis in the support of $\psi$ commute with the Pauli corresponding to $\tilde d_S=\Omega d_S$. Then for the second inequality to be an exact equality we must have that the state is pure on subsystem $S$, hence the full state is a product across $S|S^c$. Note that if the overall constant shift would have been -1, then the constraint would instead become
\begin{equation}
    -1=\frac{1}{2^{|S|}}\sum_{P\in\mathcal P_S}\text{Tr}[P\psi]^2\cdot(-1)^{p^T\cdot d_S}\geq-\frac{1}{2^{|S|}}\sum_{P\in\mathfrak{\mathcal P_S}}\text{Tr}[P\psi]^2\geq-1
\end{equation}
where equality would now only be satisfied if $p^T\cdot d_S=1$ $\forall p\in \mathcal P_S$ with $\text{Tr}[P\psi]^2>0.$ However, this case can never arise, since  $\text{Tr}[\psi]^2>0$ and $p_{\mathbb i}^T\cdot d_S=0$.
Now take the case where $d_S=0$, then
\begin{align}
    1&=\frac{1}{2^{|S|}}\sum_{P\in\mathcal P_S}\text{Tr}[P\otimes Q\psi]^2\\
    &=\frac{1}{2^{|S|}}\sum_{P\in\mathcal P_S}\text{Tr}_1[P\text{Tr}_2[\mathbb I\otimes (\Pi_+^Q-\Pi_-^Q)\psi]]^2\notag\\
    &=\frac{1}{2^{|S|}}\sum_{P\in\mathcal P_S}\text{Tr}_1[P(\psi_+^Q-\psi_-^Q)]^2\notag\\
    &=\frac{1}{2^{|S|}}\sum_{P\in\mathcal P_S}\text{Tr}[P\psi_+^Q]^2+\text{Tr}[P\psi_-^Q]^2-2\text{Tr}[P\Psi_+^Q]\text{Tr}[P\Psi_-^Q]\notag\\
    &=\text{Tr}[{\psi_+^Q]}]^2+\text{Tr}[{\psi_-^Q]}]^2-2\text{Tr}[\psi_+^Q\psi_-^Q]\notag\\
    &\leq \text{Tr}[{\psi_+^Q]}]^2+\text{Tr}[{\psi_-^Q]}]^2\leq 1\notag
\end{align}
where we have split the Pauli $Q$ into its $+1/-1$ eigenspace projections and used that the reduced state $\psi_i^Q$ is positive. For this not to become a contradiction, the inequalities need to be exactly satisfied, forcing $\psi$ to live exactly in one of the eigenspaces of $Q$, i.e., $Q\psi=\pm\psi$ and hence $Q\in\text{STAB}(\psi)$. Then, assuming without loss of generality that $\psi_-^Q=0$, the rest of the equation imposes purity of the $S$ subsystem, i.e., $\psi$ is a product across $S|S^c$. Finally, for the generic case where both $d_S\neq 0$ and $Q\neq \mathbb I$
\begin{align}
    1&=\frac{1}{2^{|S|}}\sum_{P\in\mathcal P_S}\text{Tr}[P\otimes Q\psi]^2(-1)^{p^T\cdot d_S}\\
    &\leq \frac{1}{2^{|S|}}\sum_{P\in\mathcal P_S}\text{Tr}[P\otimes Q\psi]^2\leq 1\notag
\end{align}
where equalities are only satisfied if the above conditions are satisfied, namely if both $d_S$ and $Q$ are stabilizers and the state is a product.
\end{proof}

All of these examples, showcase how some typical instances can be phrased in terms of quadratic symmetries of the Bell distribution. Finally, before moving on, note that given a state $\rho$, the set of all Bell symmetries it satisfies forms a vector space. Indeed, given $M,N$ two Bell symmetries of $\rho$
\begin{equation}
    q^TMq=c_M,\quad q^TNq=c_N\quad \forall q\in\text{supp}(\mathsf b_\rho),
\end{equation}
then clearly
\begin{equation}
    q^T(N+M)q=c_N+c_M,\quad \forall q\in\text{supp}(\mathsf b_{\rho}),
\end{equation}
which is also a constant. Then, from now on we will refer to this vector space as the quadratic kernel of the Bell distribution, i.e.,
\begin{equation}
    \text{Ker}_2(\rho)=\{M\in\mathbb F_2^{2n\times 2n}:q^TMq=c\quad\forall q\in\text{supp}(\mathsf b_\rho)\}.
\end{equation}
Then, the previous problems of stabilizer, or hidden cut testing / learning can be seen as the problem of figuring out whether some particular symmetry is contained in this vector space.

\section{Clifford dressed product states}
\label{sec:clifford_dressed}

Like we said, in this case, our goal is to understand one particular family of such structured states, namely those of the form
\begin{equation}
    \ket{\psi}=U_C(\ket{\psi_A}\otimes\ket{\psi_B})
\end{equation}
where $U_C$ is an unknown Clifford. We keep this discussion to the bipartite case for simplicity, but it follows similarly for the general case. Before applying $U_C$, the Bell distribution would be such that $\{\Omega_+^A,\Omega_+^B\}\subseteq\text{Ker}_2(\psi)$. Therefore, the main question then is whether, and how, these product symmetries remain visible after the action of $U_C,$ i.e., how does $\text{Ker}_2{(\ket{\psi})}$ get transformed.

Fortunately we can show that Clifford transformations preserve this quadratic structure. More precisely, they map any quadratic symmetry to another one by an affine transformation. We show this in  \Cref{lemma:clifford_effect_preliminaries}. We introduce one auxiliary lemma before.

\begin{lemma}[{[Transpose square of a Clifford unitary]}]
    \label{lemma:U^TU_1}
    Given $U_C$ a Clifford unitary, $U_C^TU_C$ is proportional to some Pauli W. In particular, in symplectic representation, $w=\Omega\cdot\mathrm{diag}(C^T\Omega_+C)$, where $C$ is the symplectic action of $U_C$. This measures how far $U_C$ is from being a real Clifford.
    
\end{lemma}
\begin{proof}
    Take any Pauli $P$. Since $U_C$ is a Clifford, we have $Q=U_CPU_C^\dagger$ for some other Pauli $Q$. Now
    \begin{align}
        U_C^TU_CPU_C^\dagger U_C^*&=U_C^TQU_C^*\\
            &=(U_C^\dagger Q^T U_C)^T\notag\\
            &=\pm (U_C^\dagger QU_C)^T\notag\\
            &=\pm P^T=\pm P\notag
    \end{align}
    which shows that the $U_C^TU_C$ conjugates any Pauli $P$ to itself up to a sign. We can show that any unitary $V$ satisfying this must be proportional to a Pauli. We start by writing
    \begin{equation}
        VPV^\dagger=s(P) P\iff VP=s(P)PV
    \end{equation}
    where $s(P)$ is the sign associated to $P$. Write $V$ in the Pauli basis $V=\sum_ac_aP_a$, and note that, for every term $P_a$, we have $P_aP=(-1)^{[P_a,P]}PP_a$. This gives us
    \begin{align}
        VP=\sum_ac_aP_aP&=\sum_ac_a(-1)^{[P_a,P]}PP_a=s(P)\sum_ac_aPP_a=s(P)PV.
    \end{align}
    Comparing term by term, we have that if $c_a\neq 0$ we must have $(-1)^{[P_a,P]}=s(P)$ $\forall P$. Now assume $\exists$ $a,b$ s.t. $c_a,c_b\neq 0$. Then for every Pauli $P$
    \begin{equation}
        (-1)^{[P_a,P]}=s(P)=(-1)^{[P_b,P]},
    \end{equation}
    which means that $[P_a\cdot P_b,P]=0$ $\forall P$. However, since the symplectic form is non-degenerate this means $P_a\cdot P_b=\mathbb I\iff P_a=P_b$. Hence, $V\propto W\in\mathcal P_n$. 

    In order to see the specific form for $W$, we can look at its action on a generic Pauli P. Take first $U_C$, then
    \begin{equation}
        \label{eqn:CCT_eqn1}
        U_CPU_C^\dagger=\epsilon_qQ,\quad \text{for $Q$ s.t.}\quad q=Cp.
    \end{equation}
    for some phase $\epsilon_q$.Now, taking the transpose of the above expression gives
    \begin{equation}
        \label{eqn:CCT_eqn2}
        U_C^*P^TU_C^T=\epsilon_qQ^T\iff (-1)^{p^T\Omega_+p}U_C^*PU_C^T=\epsilon_q(-1)^{q^T\Omega_+q}Q.
    \end{equation}
    Plugging \Cref{eqn:CCT_eqn1} into \Cref{eqn:CCT_eqn2} yields
    \begin{equation}
        (-1)^{p^T\Omega_+p}U_C^*PU_C^T=(-1)^{q^T\Omega_+q}U_CPU_C^\dagger,
    \end{equation}
    and thus
    \begin{equation}
        U_C^TU_CPU_C^\dagger U_C^*=(-1)^{p^T\Omega_+p+q^T\Omega_+q}P.
    \end{equation}
    Now we can reduce to 
    \begin{equation}
        U_C^TU_CPU_C^\dagger U_C^*=(-1)^{p^T\Omega_+p+(Cp)^T\Omega_+(Cp)}P.
    \end{equation}
    Then, $W$ corresponds to the Pauli whose commutation sign with $P$ is given by $(-1)^{p^T\Omega_+p+(Cp)^T\Omega_+(Cp)}$ which precisely measures the difference between the imaginarity of $P$ and its image under $U_C$. In symplectic notation,
    
    \begin{equation}
        w^T\Omega \cdot p=p^T\Omega_+p+(Cp)^T\Omega_+(Cp)=p^T(C^T\Omega_+C+\Omega_+)p.
    \end{equation}
    Since $C$ is symplectic the polar form of this quadratic expression is 
    \begin{equation}
        (C^T\Omega_+C+\Omega_+)+(C^T\Omega_+C+\Omega_+)^T=C^T\Omega C+\Omega=0
    \end{equation}
    so that this form is actually linear, which means that only the diagonal matters. Then we can rewrite
    \begin{equation}
        p^T(C^T\Omega_+ C+\Omega_+)p=\text{diag}(C^T\Omega_+C+\Omega_+)\cdot p=\text{diag}(C^T\Omega_+C)\cdot p.
    \end{equation}
    Ultimately, $w$ is the vector such that
    \begin{equation}
        w^T\Omega=\text{diag}(C^T\Omega_+C)
    \end{equation}
    so
    \begin{equation}
        w=\Omega\cdot\text{diag}(C^T\Omega_+C)^T.
    \end{equation}
\end{proof}

With this we can now show how the action of a Clifford on a quadratic Bell symmetry looks like.

\mainlemma*

\begin{proof}
    We begin by looking at the effect of a Clifford on the Bell distribution of the state $\rho,$ which unlike for the Pauli distribution the transformation is more involved.
    
    Action by a Clifford $U_C$ on a state performs a linear transformation on the Pauli support such that for $p\in\text{supp}(\mathsf p_\rho)$
    \begin{equation}
        p\mapsto Cp+r.
    \end{equation}
    
    Then let $C\in \mathrm{Sp}(2n,\mathbb F_2)$, describe the Pauli tableau of the Clifford unitary $U_C$.
The effect of this term on the Bell distribution is then the following
\begin{equation}
    \text{tr}[Q(U_C\rho U_C^\dagger)Q(U_C\rho U_C^\dagger)^T]=\text{tr}[(U_C^{ T}QU_C)\rho (U_C^\dagger QU_C^{*})\rho^T]
\end{equation}
so the support of the Bell distribution, labeled by Pauli operators, transforms according to 
\begin{equation}
    U_C^T QU_C=(U_C^TU_C)U_C^\dagger QU_C
\end{equation}
which means that they transform under the standard adjoint action of the Clifford, with an additional affine shift given by $W=U_C^TU_C$. The overall transformation in symplectic representation is then an affine transformation $q\mapsto Cq+w$, where $w$ is of the form in \Cref{lemma:U^TU_1}.

Now, let the original distribution be labeled $q$ and the Clifford conjugated one be $\tilde q$ in symplectic representation. Since both distributions are related by a Clifford, from the transformation above, we have that $\tilde q=Cq+w\iff q=C'\tilde q+w'$ (where $C'=C^{-1}$, $w'=C^{-1}w$). Assume the original distribution satisfies some symmetry $M$, i.e., $q^TM q=c$, then we have
\begin{align}
    c&=(C'\tilde q+w')^TM(C'\tilde q+w')\\
    &=\tilde q^TC'^TMC'\tilde q+w'^T(M+M^T)C'\tilde q+w'^TMw'\notag\\
    \tilde c=c+w'^TMw'&=\tilde q^TC'^TMC'\tilde q+\lambda_C^TC'\tilde q\notag\\
    &=\tilde q^TC'^T(M+\lambda_C\lambda_C^T)C'\tilde q\notag
\end{align}
where $\lambda_C\coloneqq(M+M^T)w'$ and we used $a^2=a$ for $a\in\mathbb F_2.$ This means that after Clifford conjugation, the distribution also satisfies an (affine) quadratic constraint. In particular, for state vectors of the form $\ket{\psi}=U_C\ket{\psi_A}\otimes\ket{\psi_B}$, since $\ket{\psi_A}\otimes\ket{\psi_B}$, satisfies the purity symmetries $\Omega_+^A,\Omega_+^B$, plugging this in the expression above concludes the proof.
\end{proof}

The lemma shows that for the product case, after applying an unknown Clifford, the state need no longer be product across the original bipartition and may in fact be highly entangled. Nevertheless, it remains structured at the level of Bell sampling: the symmetries are not destroyed, but are transformed into new ones, and the state then presents an equally large non-trivial quadratic kernel. 

The main caveat here is that since $C$, i.e., the affine transformation, is also unknown, the specific form of the symmetry we are looking for is also unknown. Moreover, for stabilizer or vanilla hidden cut learning, the symmetries are essentially linear. Even for the hidden cut case, the symmetry has the form
\begin{equation}
    \Omega_+^A=\bigoplus_{i\in A}\Omega_+^{(i)},
\end{equation}
which behaves linearly on the qubit registers. One can then test the local purity form on each qubit and solve a system of linear equations to retrieve A. In our setting this need not be the case due to the affine transformation, and thus we will have to deal with truly quadratic symmetries. We will now discuss the additional classical post-processing steps required to recover these structures.

\subsection{Symmetry space retrieval}

We now assume access to Bell samples of a state of the form $\ket{\psi}=U_C\ket{\psi_A}\otimes\ket{\psi_B}$, we will later on discuss how many samples are required. Like we said, this distribution still satisfies some non-trivial quadratic symmetries. In particular, at least one of these corresponds to an affine transformed subsystem purity symmetry $\Omega_+^A$, for some symplectic transformation $C.$ 

The first step then is to recover a basis for the full quadratic kernel $\text{Ker}_2(\rho)$ from the Bell samples. The key observation in order to do this, is that we can turn the quadratic symmetries into linear ones, by lifting the Bell labels to a symmetric subspace, i.e., let

\begin{align}
    \phi:\mathbb F_2^{2n}&\rightarrow\mathbb F_2^{2n}\otimes \mathbb F_2^{2n}\\
        p&\mapsto p^{\otimes 2}\notag
\end{align}
and $|M\rangle\rangle\in \mathbb F_2^{2n\times 2n}$ be the vectorization of some quadratic symmetry $M$. Then,
\begin{equation}
    q^TMq=0\iff \langle\langle M|\phi(q)\rangle=0,
\end{equation}
which is a linear condition. Note that for non-zero affine shifts, we can have
\begin{equation}
    \langle\langle M|\phi(q)\rangle=1,
\end{equation}
we can fix some $q_0$ and look at the difference distribution, i.e.,
\begin{equation}
   \langle\langle M|\phi(q)+\phi(q_0)\rangle=0.
\end{equation}
With this in mind, we can then define the lifted Bell distribution by
\begin{equation}
    \tilde {\mathsf b}_\rho=\{\phi(q)+\phi(q_0):q,q_0\in\supp(\mathsf b_\rho)\}\subseteq\mathbb F_2^{2n}\otimes \mathbb F_2^{2n}.
\end{equation}
While this does not form a closed subspace, we can define its linear closure and look at this instead
\begin{equation}
    V_\rho=\text{span}_{\mathbb F_2}\tilde{\mathsf b}_\rho.
\end{equation}
The reason why we can do this is that at this stage we do not strictly care about the Bell distribution, but only about the symmetries it satisfies. It can be seen that if all elements of $\tilde {\mathsf b}_\rho$ satisfy some symmetry $M$, then this is also a symmetry for all of $V_\rho$, precisely because in this lifted space the symmetries behave linearly. Hence, the following equality is satisfied
\begin{equation}
     \text{Ker}_2(\rho)=V_\rho^{\perp},
\end{equation}
i.e., the quadratic kernel of the Bell distribution is equal (up to vectorization of the symmetries) to the linear kernel of the lifted Bell closure. Thus, given a basis of $V_\rho$ one can obtain a basis of $\text{Ker}_2(\rho)$ by standard Gaussian elimination over $\mathbb F_2.$

Given this basis, the task is then to find a symmetry inside $\text{Ker}_2(\rho)$ satisfying
\begin{equation}
    M = C^{-T}(\Omega_+^A+d)C^{-1}
\end{equation}
for some $C, A$ and $d$, and learning these, for the testing and learning tasks respectively.

In the next section we will show how to do this given a naive algorithm running in exponential time, and show that the number of Bell samples required for it to work remains bounded, and later lift it to an algorithm running in polynomial time.

\subsection{Sample complexity I: exponential runtime algorithm}

We just showed that, given access to the Bell distribution, the task can be reduced to finding an element of the particular form above in the kernel. However, in practice we are given a finite number of Bell samples in order to construct this kernel. From now on, let us call this $\text{Ker}_{2,N}(\rho)$, i.e., the observed kernel when given access to $N$ Bell samples of $\rho.$

It could happen then, that $\text{Ker}_{2,N}(\rho)$ is too big, for example because for some $M$, the Bell labels such that $q^TMq\neq c$ had too low probability and were not sampled. Then, we will need to make sure that for $N=O(\text{poly}(n))$, $\text{Ker}_{2,N}(\rho)$ is a good enough approximation of the true kernel. We will first show this for the setting where we use a naive algorithm with exponential runtime. 

We will restrict to the pure state case and describe the procedure of \Cref{alg:exp_runtime}. Given a basis of $\text{Ker}_{2,N}(\psi)$, the algorithm iterates over all symplectic matrices $C\in \text{Sp}(2n,\mathbb F_2)$. For each such $C$, we compute the conjugate
\begin{equation}
    M_i(C):=C^{-T}\Omega_+^{(i)}C^{-1},\quad i\in[n],
\end{equation}
as well as diagonal contributions $C^{-T}dC^{-1}$, for $d\in\mathbb F_2^{2n}$. Then, since for any subset $S\subseteq [n]$ one has
\begin{equation}
    \Omega_+^S=\sum_{i\in S}\Omega_+^{(i)},
\end{equation}
deciding whether there exists $\emptyset\neq S\subseteq[n]$ such that 
\begin{equation}
    C^{-T}(\Omega_+^S+d)C^{-1}=\sum_{i=1}^ns_iM_i(C)+C^{-T}dC^{-1}\in\text{Ker}_{2,N}(\psi),
\end{equation}
for $\vec s \in\mathbb F_2^n$, can be done by solving a linear system of equations. The solution to this will yield all possible cuts hidden by that $C.$ Having done this for all $C$'s we keep the choice of $C$ and $S$'s that yields the maximal number of cuts. Overall, the runtime of this algorithm is dominated by the iteration over $\text{Sp}(2n,\mathbb F_2)$, and scales roughly like $O(2^{n^2})$.

We now show that for such an algorithm, with polynomially many samples, we can construct a good enough kernel, such that if such an $M$ is found, then the state is of the given form with high probability, and hence the problems in Definition 1 and 2 can be solved with bounded number of samples. Before doing so, however, let us make a remark regarding the choice of the distance measure between states and the notation used.

\begin{remark}[{[Choice of distance measure]}]
    \label{remark:distance_measure}
    The relevant distances between pure states in \Cref{def:testing_task} and \Cref{def:learning_task} are expressed using the Euclidean $2$-norm on state vectors. In a slight abuse of notation, since a global phase is physically irrelevant, this is understood in the phase-minimized sense
    \begin{equation}
        d_2(\psi,\phi):=\min_{\theta\in[0,2\pi)}\left\|\ket{\psi}-e^{i\theta}\ket{\phi}\right\|_2
        =\sqrt{2-2|\langle\psi|\phi\rangle|}.
    \end{equation}
    We use this distance because it interacts particularly well with the vector-level product decompositions constructed by our algorithms. For pure states it is equivalent, up to universal inequalities, to the operationally more standard trace distance
    \begin{equation}
        D_{\rm tr}(\psi,\phi)=\frac12\left\|\ketbra{\psi}{\psi}-\ketbra{\phi}{\phi}\right\|_1
        =\sqrt{1-|\langle\psi|\phi\rangle|^2},
    \end{equation}
    with
    \begin{equation}
        D_{\rm tr}(\psi,\phi)\leq d_2(\psi,\phi)\leq \sqrt{2}\,D_{\rm tr}(\psi,\phi).
    \end{equation}
    It therefore leads to the same notion of efficient approximate learning while keeping the vector-level analysis transparent.
\end{remark}

\begin{lemma}[{[Exponential-time Clifford-product testing]}]
    \label{lemma:sample_test_exp}
    The testing task in \Cref{def:testing_task} can 
    be solved in exponential time using 
    \begin{equation}N=O\bigg(\frac{1}{\epsilon^2}\big(n^2+\log(\frac{1}{\delta})\big)\bigg)
    \end{equation}
    many copies of the state. 
\end{lemma}

\begin{proof}
    For states satisfying the condition of 
    Case $A$ the argument is trivial since for any number of samples 
\begin{equation}
    \mathsf p(M\in\text{Ker}_{2,N}(\psi))=1
\end{equation}
for $M=C^T(\Omega_+^A+d)C$, and after performing the exhaustive search the algorithm will test positive. For states satisfying Case $B$ it could happen that a lack of samples leads to some such $M$ being accidentally identified. The probability of this happening however is bounded by the purity of the conjugated $S$ subsystem. Let us look first at the case $M=C^T(\Omega_+^A+\lambda_C)C$, i.e., when the diagonal matches the $C$ induced shift we have that

\begin{equation}
    \label{eqn:purity_probability}
    \mathsf p(M\in\text{Ker}_{2,2}(\psi))=\frac{1+\text{Tr}[\tilde\rho_A^2]^2}{2},\qquad\tilde\rho_{A}=\text{Tr}_{B}[U_C^{\dagger}\ket{\psi}\!\!\bra{\psi}U_C].
\end{equation}

Setting aside the effect of $U_C$,  we have already shown in \Cref{lemma:bell_to_Pauli} and \Cref{eqn:subsystem_purity_symm}, that the purity of some subsystem $A$ is related to $\Omega_+^A$ by
\begin{equation}
    \text{Tr}[\psi_A^2]=\sum_{Q\in\mathcal P_n}\mathsf b_{\psi}(Q)\cdot(-1)^{q^T\Omega_+^Aq}=\mathbb E_{q\sim\mathsf b_{\psi}}[(-1)^{q^T\Omega_+^Aq]}]
\end{equation}
from which it follows that
\begin{equation}
    \mathsf p_{q\sim\mathsf b_\psi}(\Omega_+^A\in\text{Ker}_{2,2}(\psi))=\frac{1+\mathbb E_{q\sim\mathsf b_{\psi}}[(-1)^{q^T\Omega_+^Aq]}]}{2}=\frac{1+\text{Tr}[\psi_A^2]}{2}.
\end{equation}
Looking at the difference distribution with respect to some fixed reference $q_0$ gives us
\begin{equation}
    \mathsf p_{q\sim\mathsf b_\psi}(\Omega_+^A\in\text{Ker}_{2,2}(\psi)|q_0)=\frac{1+(-1)^{q_0^T\Omega_+^Aq_0}\mathbb E\left[(-1)^{q^T\Omega_+^Aq}\right]}{2}=\frac{1+(-1)^{q_0^T\Omega_+^Aq_0}\text{Tr}[\psi_A^2]}{2}.
\end{equation}
Since conditioned on $q_0$, the acceptance events of $N$ many draws are independent, it holds that
\begin{equation}
    \label{eqn:purity_probabilityN}
    \mathsf p_N(\Omega_+^A\in\text{Ker}_{2,N}(\psi)|q_0)\leq\left(\frac{1+\text{Tr}[\psi_A^2]}{2}\right)^N
\end{equation}
which also holds after averaging over $q_0.$ By defining the corresponding subsystem $\tilde\rho_A$ accordingly, i.e., taking into account the Clifford action, the same argument gives \Cref{eqn:purity_probability,eqn:purity_probabilityN} for $M=C^{-T}(\Omega_+^A+\lambda_C)C^{-1}$. On the other hand, whenever the diagonal is not chosen appropriately, i.e., $M=C^{-T}(\Omega_+^A+d)C^{-1}$, for $d\neq\lambda_C$, it follows from \Cref{lemma:diagonal_stab} that 
\begin{equation}
    \mathsf p(M_d\in\text{Ker}_{2,2}(\psi))\leq p(M\in\text{Ker}_{2,2}(\psi)),
\end{equation}
and hence the same bound holds.

Now, if $\lambda_{\max}^A$ is the maximum eigenvalue of $\tilde\rho_A$, then
\begin{equation}
   \min_{\psi_A,\psi_B} ||U_C^\dagger \ket{\psi}-\ket{\psi_A}\otimes\ket{\psi_B}||_2^2=2-2\sqrt{\lambda^A_{\text{max}}}.
\end{equation}
Given that
\begin{equation}
    \forall U_C, \forall A\subset[n]\backslash\emptyset: \|U_C^\dagger\ket{\psi}-\ket{\psi_A}\otimes \ket{\psi_B}\|_2\geq \epsilon
\end{equation}
this means
\begin{equation}
    \lambda^A_{\text{max}}\leq\left(1-\frac{\epsilon^2}{2}\right)^2.
\end{equation}
However
\begin{equation}
    \text{Tr}[\tilde\rho_A^2]=\sum_i{\lambda^A_i}^2\leq \lambda^A_{\text{max}}\sum_i\lambda^A_i=\lambda^A_{\text{max}}\leq\left(1-\frac{\epsilon^2}{2}\right)^2,
\end{equation}
and thus after taking $N$ many samples
\begin{equation}
    \mathsf p_N(M\in\text{Ker}_{2,N}(\psi))\leq\big(1-\frac{3\epsilon^2}{8}\big)^N\leq  e^{-3N\epsilon^2/8}.
\end{equation}
Finally, we can roughly approximate the number of all candidate symmetries by
\begin{equation}
    |\mathcal M|\leq|\text{Sp}(2n,\mathbb F_2)|\times2^{3n}=2^{O(n^2)},
\end{equation}
which we can then use to perform a union-bound
\begin{equation}
    \mathsf p_N(\text{wrong}\ M)\leq |\text{Sp}(2n,\mathbb F_2)|\times2^n\cdot e^{-3N\epsilon^2/8}
\end{equation}
Setting the maximum error probability to $\delta$, it is then enough to use
\begin{equation}
    N=O\bigg(\frac{1}{\epsilon^2}\big(n^2+\log\left(\frac{1}{\delta}\right)\big)\bigg)
\end{equation}
many samples, so the number of samples required to make the probability to sample the wrong symmetries exponentially small remains efficient.
\end{proof}

\begin{lemma}[{[Exponential-time Clifford-product learning]}]
    \label{lemma:sample_learn_exp}
    The learning task in \Cref{def:learning_task} can be solved in exponential time using $N=O\bigg(\frac{n}{\epsilon^2}\big(n^2+\log(\frac{1}{\delta})\big)\bigg)$ many copies of the state.
\end{lemma}

\begin{proof}
    The proof proceeds similarly to the case before. In this case we need to make sure that we do not sample enough wrong symmetries such that the disentangled description fails to approximate the original state to accuracy $\epsilon$. 
    
    Let $(U_C,\{C_i\}_{i=1}^m)$ be a Clifford and a partition such that
    \begin{equation}
        \label{eqn:distance_guarantee}
        \left|\left|\ket{\psi}-U_C\bigotimes_{i=1}^m\ket{\phi_i}_{C_i}\right|\right|_2>\epsilon,
    \end{equation}
    for any choice of $\{\phi_i\}_{i=1}^m$. We can bound the probability of obtaining such a description in a similar way as before.
    
    Indeed, let
    \begin{equation}
        \rho_i^C=\text{Tr}_{C_i^c}[U_C^\dagger\ket\psi\!\!\bra\psi U_C].
    \end{equation}
    Just like before, the probability of observing the corresponding symmetry $M_i=C^{-T}\Omega_+^{i}C^{-1}$ (we ignore here diagonal affine shifts) is given by
    \begin{equation}
        \mathsf p_N(M_i\in\text{Ker}_{2,N}(\psi))\leq\left(\frac{1+\text{Tr}[(\rho_i^C)^2]}{2}\right)^N.
    \end{equation}
    Then we can show that given \Cref{eqn:distance_guarantee}, there is at least one $i$ for which $\mathsf p(M_i\in\text{Ker}_{2,N}(\psi))\leq1-\eta$. 
    Assume the contrary, i.e., every such block has high purity, such as
\begin{equation}
        \text{Tr}[(\rho_i^C)^2]>1-\eta,\quad\forall i.
    \end{equation}
    Let $\lambda_i$ be the largest eigenvalue of $\rho_i^C$. Then
    \begin{equation}
        \text{Tr}[(\rho_i^C)^2]\leq \lambda_i.
    \end{equation}
    Now let $\ket{u_i}$ be the top eigenvector of $\rho_i^C$ and its associated projector $\Pi_i=\ket{u_i}\!\!\bra{u_i}$. Then
    \begin{equation}
        \bra\psi U_C\Pi_iU_C^\dagger\ket\psi=\lambda_i>1-\eta.
    \end{equation}
    Since the different $\Pi_i$'s act on different blocks, they commute, and
    \begin{equation}
        \langle\psi |U_C\prod_{i=1}^m\Pi_iU_C^\dagger|\psi\rangle\geq1-\sum_{i=1}^m(1-\bra\psi U_C\Pi_i U_C^\dagger\ket\psi)>1-m\eta.
    \end{equation}
    However,
    \begin{equation}
        \prod_{i=1}^m\Pi_i=\bigotimes_{i=1}^m\ket{u_i}\!\!\bra{u_i},
    \end{equation}
    so
    \begin{equation}
        \left|\langle\psi| U_C\bigotimes_{i=1}^m|u_i\rangle\right|^2>1-m\eta,
    \end{equation}
    and equivalently
    \begin{equation}
        \left|\left|\ket{\psi}-U_C\bigotimes_{i=1}^m\ket{u_i}_{C_i}\right|\right|_2\leq \sqrt{2-2\sqrt{1-m\eta}}.
    \end{equation}
    On the other hand, from our initial claim, 
    we must have that
    \begin{equation}
        \epsilon<\sqrt{2-2\sqrt{1-m\eta}}\iff \eta>\frac{\epsilon^2-\epsilon^4/4}{m}.
    \end{equation}
    This means that for $\eta=\frac{\epsilon^2-\epsilon^4/4}{m}$, there is some $i$, such that 
    \begin{equation}
        \text{Tr}[(\rho_i^C)^2]\leq\bigg(1-\frac{\epsilon^2-\epsilon^4/4}{m}\bigg)\approx1-\frac{\epsilon^2}{m},
    \end{equation}
    and then
    \begin{equation}
        \mathsf p_N(M_i\in\text{Ker}_{2,N}(\psi))\leq \left(1-\frac{\epsilon^2}{2m}\right)^N\leq e^{-\frac{N\epsilon^2}{2m}}.
    \end{equation}

    So the probability of observing symmetries leading to wrong descriptions of the state ($\epsilon$ approximation) gets exponentially suppressed with the number of samples. Then union bounding over the total number of possible wrong symmetries, we roughly get
    \begin{equation}
        \mathsf p_N(\text{wrong}\ M)\leq 2^{O(n^2)}\cdot e^{-\left(\frac{N\epsilon^2}{2m}\right)},
    \end{equation}
    where we have bounded the number of wrong purity symmetries by 
    \begin{equation}
        |\mathcal M_{wrong}|\leq|\text{Sp}(2n,\mathbb F_2)|\times 2^{3n}=2^{O(n^2)}.
    \end{equation}
    Then for a maximum error probability $\delta$, it is enough to use 
    \begin{equation}
        N\geq\frac{m}{\epsilon^2}\bigg(\log(|\mathcal M_{wrong}|)+\log(\frac{1}{\delta})\bigg)
    \end{equation}
    which since $m\leq n$, gives
    \begin{equation}
        N=O\bigg(\frac{n}{\epsilon^2}\big(n^2+\log(\frac{1}{\delta})\big)\bigg)
    ,\end{equation}
    Then, for such $N$, the probability to observe symmetries yielding some description $U_C\bigotimes_{i=1}^m\ket{\phi_i}_{C_i}$ that is $\epsilon$ far from the true $\ket{\psi}$ in 2-norm is exponentially suppressed. Then, for any set of transformed purities that the algorithm finds, we can guarantee that the description this will yield is $\epsilon$ close to the true state.
\end{proof}

Having shown this, we will now discuss the computationally efficient protocols and their sampling complexities.

\begin{algorithm}[ht]
\label{alg:exp_runtime}
\caption{Exponential runtime algorithm for task in \Cref{def:learning_task}.}
\label{alg:exponential_algorithm}
\begin{algorithmic}[1]

\Statex \textbf{Input:} 
\(2N=O(\frac{n}{\epsilon^2}(n^2+\log(1/\delta)))\) copies of some state vector $\ket{\psi}$.

\Statex \textbf{Promise:} There exist a Clifford $U_C$, and some partition $\{c_i\}_{i=1}^m$ such that for some $\{\psi_i\}_{i=1}^m$ it holds that $\ket{\psi}=U_C\bigotimes_{i=1}^m\ket{\psi_i}_{c_i}$.

\Statex \textbf{Output:} 
A Clifford $U_{\tilde C}$ and a partition $\{\tilde{ c}_i\}_{i=1}^{m'}$ for $m'\geq m$ such that $||U_{\widetilde C}\bigotimes_{i=1}^{m'}\ket{\phi_i}_{\tilde C_i}-\ket\psi||_2\leq \epsilon$ for some $\{\phi_i\}_{i=1}^{m'}.$

\State Perform Bell sampling on the \(2N\) copies of \(\ket{\psi}\) and construct a basis for $\text{Ker}_{2,N}(\psi)$.

\ForAll{\(C\in \text{Sp}(2n,\mathbb F_2)\)}
    \For{\(i=1,\dots,n\)}
        \State Compute and store $M_i(C):=C^T\Omega_+^{(i)}C$.
        
    \EndFor

    \State Solve the linear system 
    $\sum_{i=1}^n s_i M_i(C)+C^T dC
    \in
    \text{Ker}_{2,N}(\psi)$,
    
    with unknowns $\mathbf s=(s_1,\dots,s_n)\in\mathbb F_2^n$ and $\mathbf d\in\mathbb F_2^{2n}$.
\EndFor
\State \textbf{return} $U_{\tilde C}$ \textit{and} $\{\tilde{c}_i\}_{i=1}^{m'}$ \textit{that yield the largest} $m'$.

\end{algorithmic}
\end{algorithm}

\section{Computationally efficient algorithm}
\label{sec:comp_eff}
Above, we have shown that the problem can be solved using polynomially many samples, by phrasing it as a symmetry search problem, albeit at an exponential cost in runtime. Here we will show that an efficient algorithm running in polynomial time also exists, and again works given only polynomially many copies of the state.
In order to show this, we will need to introduce a key fact about Bell symmetries of product states. We will show, that for some state vector $\ket{\psi}=\ket{\psi_A}\otimes\ket{\psi_B}$, not only $\Omega_+^A,\Omega_+^B\in\text{Ker}_2(\psi)$, but also, for any other symmetry $M\in\text{Ker}_2(\psi)$, this needs to respect the cut.

\productblock*

\begin{proof}

Let
\begin{equation}
    S_A=\supp(\mathsf b_{\psi_A}),
    \qquad
    S_B=\supp(\mathsf b_{\psi_B}).
\end{equation}
Since $\ket{\psi}$ is a product state vector, its Bell distribution factorizes
\begin{equation}
    \mathsf b_\psi(P_A\otimes P_B)
    =
    \mathsf b_{\psi_A}(P_A)\,
    \cdot\mathsf b_{\psi_B}(P_B),
\end{equation}
and hence the support of $\mathsf b_\psi$ satisfies
\begin{equation}
    \supp(\mathsf b_\psi)
    =
    S_A\otimes S_B.
\end{equation}
Let us then write the samples as $q=(q_A,q_B)\in\mathbb F_2^{2|A|}\times \mathbb F_2^{2|B|}$, and take arbitrary $(q_A,q_B),(q_A',q_B')\in\text{supp}(\mathsf b_\psi)$. Due to the factorisation of $\text{supp}(\mathsf b_\psi)$, we can guarantee that $(q'_A,q_B),(q_A,q_B')\in\text{supp}(\mathsf b_\psi)$. Now let $M$ be some symmetry of $\ket{\psi}$, of the form
\begin{equation}
    M
    =
    \begin{pmatrix}
        M_A & M_{A, B}\\
        0   & M_B
    \end{pmatrix}.
\end{equation}
We will show that $M_{A, B}=0$. Let $Q_M(q):=q^TMq$. Then, for the 4 samples above one can check that
    \begin{equation}
        \label{eqn:cross_term}
        Q_M(q_A,q_B)+Q_M(q'_A,q'_B)+Q_M(q'_A,q_B)+Q_M(q_A,q'_B)=(q_A+q'_A)^TM_{A, B}(q_B+q_B').
    \end{equation}
Since all four terms are contained in the support, they must satisfy the constraint $M$, and hence for all four samples $Q_M$ has the same constant value $c.$ Thus, \Cref{eqn:cross_term} reduces to
\begin{equation}
    (q_A+q'_A)^TM_{A, B}(q_B+q_B')=0.
\end{equation}
Now suppose for contradiction that M is not block-diagonal, i.e., $M_{A, B}\neq 0.$ If there exists $q_B,q'_B\in S_B$ such that
\begin{equation}
    v_A:=M_{A, B}(q_B+q_B')\neq 0,
\end{equation}
then
\begin{equation}
    (q_A+q_A')\cdot\quad v_A=0, \forall q_A,q'_A\in S_A\iff q_A\cdot v_A=c_A\quad \forall q_A\in S_A.
\end{equation}
This is a non-trivial linear constraint on $S_A$, hence a stabilizer-type constraint, contradicting the assumption.

Otherwise,
\begin{equation}
    M_{A, B}(q_B+q_B')=0,\quad\forall q_B,q'_B\in S_B.
\end{equation}
But again, since $M_{A, B}\neq0$, there is some row of $M_{A, B}$ defining a non-zero vector $v_B$ such that
\begin{equation}
    (q_B+q'_B)\cdot v_B=0\quad\forall q_B,q'_B\in S_B\iff q_B\cdot v_B=c_B\quad\forall q_B\in S_B,
\end{equation}
which is now a non-trivial linear constraint on $S_B$, contradicting again the assumption. Hence
\begin{equation}
    M_{A, B}=0.
\end{equation}
Therefore, every quadratic symmetry is block-diagonal across the cut, i.e.,
\begin{equation}
    M=M_A\oplus M_B.
\end{equation}
Finally, by taking $(q_A,q_B)$ and $(q_A',q_B)$ in $\text{supp}(\mathsf b_\psi)$
\begin{equation}
    Q_M(q_A,q_B)+Q_M(q_A',q_B)=Q_{M_A}(q_A)+Q_{M_A}(q'_A)=0\iff Q_{M_A}(q_A)=k_A,\quad\forall q_A\in S_A
\end{equation}
and similarly for $M_B.$ Thus, this also means that the two blocks $M_A,M_B$ define local quadratic symmetries on the factors themselves. This proves the claim.

\end{proof}
The main idea of \Cref{thm:product_block} is that for product states, up to the presence of stabilizers, not only do we have that the subsystems are pure, but also that all other quadratic symmetries of $\psi$ need to stem from symmetries of the respective subsystems.

It is important to note that we can always assume this to be the case, namely the symmetries to be block-diagonal. That is because we can always identify the linear constraint by performing some stabilizer learning procedure on the Bell samples first. Having done that we can always find a Clifford $U_{C'}$ that maps the stabilizers onto some local subsystem, such that
\begin{equation}
    U_{C'}\ket{\psi}=\ket{ \psi'}\otimes\ket{0}^{\otimes k}.
\end{equation}
Then, if $\ket{\psi}$ had some other hidden cut besides these trivial \textit{stabilizer-like} ones, $\ket{\psi'}$ also does.
Now we show how this fact carries over to the actual setting that we care about, where a Clifford follows the product state.

\blockdiagonalisable*

\begin{proof}
    By the \Cref{thm:product_block} all symmetries of $\ket{\psi_A}\ket{\psi_B}$ are block diagonal across $A|B$. Then from \Cref{lemma:clifford_effect_preliminaries}, after applying $U_C$ on the state, all symmetries are of the form $M=(C^{-1})^T(M_A\oplus M_B+\lambda\lambda^T)(C^{-1})$, where the $\lambda\lambda^T$ term is diagonal and does not affect the block-structure. Then $C^TMC=M'_A\oplus M'_B$.
\end{proof}

This is the crucial point we rely upon to show that both testing and learning can be done computationally efficiently. We will first show how an algorithm running in polynomial time can exploit \Cref{cor:block-diagonalisable} to recover $C$ and the corresponding partition, in case it exists. Later, we will show again that polynomially many samples are enough to guarantee a kernel satisfying these properties.

\subsection{Simultaneous block-diagonalisation of the kernel}

Our algorithm will exploit the properties above, and recover $U_C$ and $\{C_i\}$ (in case they exist) by block-diagonalising $\text{Ker}_2(\psi)$. In order to block-diagonalise a set of quadratic forms $\mathcal M=\{M_i\}_{i=1}^{d_{\rm Ker}}$, we will again rely on the ability to block-diagonalise their polar forms $\mathcal B=\{B_i\}_{i=1}^{d_{\rm Ker}}$. Even if that implies losing information about $\text{diag}(M_i)$, by \Cref{lemma:diagonal_stab} we can always assume this to be 0 since $\ket{\psi}$ can be assumed to have no stabilizers. 

For the task that we now want to solve however, we will need to resort to simultaneous block-diagonalisation of $\mathcal B$, which asks for a single change of basis that brings every member of $\mathcal B$ into the same block structure. This is substantially more restrictive than block-diagonalising the forms one at a time: the latter may select incompatible bases and therefore reveal no common decomposition. 
Take for instance $B_1=\Omega\in\mathcal B$, which is always the polar form of some symmetry in $\text{Ker}_2(\psi)$ for pure states $\ket{\psi}$, and $B_2\in\mathcal B$ some arbitrary bilinear form. Clearly, $B_1$ is already block-diagonal, so $A_1=\mathbb I_{2n}$. On the other hand, in general, $A_2\in GL(2n,\mathbb F_2)$ block-diagonalising $B_2$, will be such that
\begin{equation}
    A_2\Omega A_2^T\neq V_A\oplus V_B.
\end{equation}

Equivalently, after converting the forms into linear operators, a simultaneous block structure corresponds to a decomposition of the underlying vector space into subspaces invariant under the algebra generated by all of these operators. This connects SBD to the classical theory of modules and finite-dimensional associative algebras, and makes common invariant subspaces---rather than the spectral decomposition of any individual matrix---the relevant objects \cite{chistov1997polynomial,ronyai1990computing,wilson2008finding}. In our setting there is an additional constraint: the common change of basis must ultimately be symplectic, because only then does it correspond to a Clifford transformation. Indeed we require an operation acting by congruence such that
\begin{equation}
    C^TB_iC=B_{i,A}\oplus B_{i,B},\quad C\in \text{Sp}(2n,\mathbb F_2).
\end{equation}

This precise constraint however, allows us to explicitly map our problem on quadratic forms to one of linear operators, thus enabling us to use the tools in \cite{chistov1997polynomial,ronyai1990computing,wilson2008finding} to perform SBD efficiently. We show this step in the following lemma.

\begin{lemma}[{[Congruence-to-similarity reduction]}]
    \label{lemma:bilinear_to_linear}
    Given $\mathcal B=\{B_i\}_{i=1}^{d}$ be some set of bilinear forms acting on $K=\mathbb F_2^{2n}$, let $\mathcal T=\{T_i\}_{i=1}^d$ be of the form
    \begin{equation}
        T_i=\Omega B_i,\quad i\in[d].
    \end{equation}
    Then, given $C\in \text{Sp}(2n,\mathbb F_2)$ such that
    \begin{equation}
        C^{-1}T_iC=V^i_A\oplus V^i_B, \quad\forall T_i\in \mathcal T,
    \end{equation}
    then 
    \begin{equation}
        C^{T}B_iC=W^i_A\oplus W^i_B, \quad\forall B_i\in \mathcal B.
    \end{equation}
\end{lemma}
\begin{proof}
    Since $C\in\text{Sp}(2n,\mathbb F_2)$, it satisfies
    \begin{equation}
        C\Omega C^T=\Omega.
    \end{equation}
    Then $\forall i\in[d]$ we have
    \begin{align}
        V_A^i\oplus V_B^i&=C^{-1}T_iC\\
        &=C^{-1}\Omega B_iC\notag\\
        &=C^{-1}C\Omega C^TB_iC\notag\\
        &=\Omega C^TB_iC,\notag
    \end{align}
    and thus
    \begin{equation}
        C^TB_iC=(\Omega_AV_A^i)\oplus(\Omega_BV_B^i).
    \end{equation}
\end{proof}

The reason \Cref{lemma:bilinear_to_linear} is useful is because by considering
the set $\mathcal T$ we can look at SBD via similarity transformation ($C^{-1}(\cdot)C$) instead of congruence ($C^T(\cdot)C$), which allows us to use some tools for finite algebras over finite fields. Indeed, we can now enlarge $\mathcal T,$ which is a vector space, in order to consider the full algebra $\mathfrak t=\text{span}_{\mathbb F_2}\langle T_i\rangle_{i=1}^{d_{\rm Ker}}$. This can be done since given a basis where any two $T_i,T_j\in\mathcal T$ are block-diagonal, then trivially $T_i+T_j$ and $T_i\cdot T_j$ are also block-diagonal. 

By considering $\mathfrak t$, the block-diagonalisation problem becomes a question about finding an irreducible decomposition of $K=\mathbb F_2^{2n}$ into minimal $\mathfrak t-$invariant subspaces $\{V_i\}$, which can be solved by explicitly constructing their corresponding projectors $\{\pi_i\}$. This set of projectors forms a subset of the \textit{centralizer} algebra $\mathfrak e$, i.e.,
\begin{equation}
    \label{eqn:centralizer}
    \{\pi_i\}\subseteq\mathfrak e\equiv \text{End}_{\mathfrak t}(K)=\{\phi\in\text{End}_{\mathbb F_2}(K):\phi a=a\phi,\forall a\in\mathfrak t\},
\end{equation}
which is itself an associative algebra, and can be efficiently constructed by solving a small system of linear equations. Since $\mathfrak e$ forms an associative algebra, by the Wedderburn-Malcev theorem \cite{ronyai1990computing}, we know that it admits a unique decomposition into a semisimple part and a nilpotent radical
\begin{equation}
    \mathfrak e=S_1\oplus\dots\oplus S_m\oplus R,
\end{equation}
where each $S_i\cong M_{n_i}(D_i) $ for some division ring $D_i$, and $R$ is a nilpotent subalgebra, i.e., $R^k=0$, for some finite $k$. We will reduce the search of projectors to the semisimple part, since $R$ contains no non-trivial idempotents. Constructing a basis for $R$ can be done efficiently for $2n\times 2n$ matrices over $\mathbb F_2$ by using the algorithm given in Ref.~\cite{ronyai1990computing}.
 With this, it is efficient to construct a subalgebra $\bar{\mathfrak e}\equiv\mathfrak e/R\cong S_1\oplus\dots\oplus S_m$. We will now show that each such simple component $S_i$ induces a projector $\pi_{S_i}$, and that in fact, in order to find the \textit{subsystem-induced} blocks that we are interested in, it is enough to compute these.

\begin{lemma}[{[Simple components recover hidden cuts]}]
    \label{lemma:simple_comp}
    Given $\mathfrak e$ defined in \Cref{eqn:centralizer}, and $\bar{\mathfrak e}$ its semisimple component. Then, each simple component $S_i\subseteq\bar{\mathfrak e}$ induces some projector $\pi_{S_i}$. Moreover, given some state vector $\ket{\psi}=U_C\ket{\psi_A}\otimes\ket{\psi_B}$, there is some $\vec i,\vec j$ such that the projectors
    \begin{equation}
        \pi_{\vec i}=\bigoplus_{\vec i}\pi_{S_i},\quad \pi_{\vec j}=\bigoplus_{\vec j}\pi_{S_j}
    \end{equation}
    can be lifted to projectors $\pi^C_{A},\pi^C_B$ onto the blocks defined by the hidden cut $A|B$.

\end{lemma}

\begin{proof}
    Let us first take the identity element $1_{\bar{\mathfrak e}}$ and write it as
    \begin{equation}
        1_{\bar{\mathfrak e}}=1_{S_1}+\dots+1_{S_m}
    \end{equation}
    with $1_{S_i}$ its restriction to each simple component. Since for $i\neq j$
    \begin{equation}
        S_iS_j\subseteq S_i\cap S_j=\{0\},
    \end{equation}
    then $\forall \phi\in S_j$,
    \begin{equation}
        \phi=1_{\bar{\mathfrak e}}\phi=\sum_i1_{S_i}\phi=1_{S_j}\phi.
    \end{equation}
    And in particular, taking $\phi=1_{S_j}$, we obtain
    \begin{equation}
        1_{S_j}^2=1_{S_j}.
    \end{equation}

    We will now see that these are enough to capture the purity blocks. For this we need to define first the center of $\bar{\mathfrak e}$, which is the subalgebra of elements commuting with every element of the algebra, i.e.,
    \begin{equation}
      Z_\mathfrak t(K)\equiv\{\phi\in\bar{\mathfrak e}|\phi\gamma=\gamma\phi,\forall\gamma\in\bar{\mathfrak e}\}.
    \end{equation}
    Crucially $Z_\mathfrak t(K)\subseteq \bar{\mathfrak e}$ is also an associative algebra, which decomposes as
    \begin{equation}
        Z_{\mathfrak{t}}(K)=Z(S_1)\oplus \dots\oplus Z(S_m),
    \end{equation}
    and therefore, projectors onto the simple components of $Z_{\mathfrak{t}}(K)$ also project onto those of $\bar{\mathfrak e}$. The important property of the center is that these projectors are now complete, in the sense that no other element $e\in Z_{\mathfrak{t}}(K)$ satisfies $e^2=e$. Indeed, each $S_i$ is isomorphic to a full matrix algebra $M_{n_i}(D_i)$, and thus, only the scalar field commutes with all elements in it. Then, since $Z(S_i)$ is a field, there can't be further non-trivial idempotents.

    Then, take some state vector $\ket{\psi}=U_C\ket{\psi_A}\otimes\ket
    \psi_B$, so that all symmetries are block-diagonalisable across $A|B$. Then there is some $M\in\text{Ker}_2(\psi)$ which is a transformed subsystem purity, and hence its polar form satisfies
    \begin{equation}
        M+M^T=C^{-T}\Omega^AC^{-1}.
    \end{equation}
    After converting to linear operators this implies that $\exists T\in\mathfrak t$ of the form
    \begin{align}
        T&=\Omega\cdot C^{-T}\Omega^AC^{-1}\\
        &=CC^{-1}\Omega C^{-T}\Omega^AC^{-1}\notag\\
        &=C(\mathbb I^A\oplus 0^B) C^{-1},\notag
    \end{align}
    i.e., which satisfies $T^2=T$, and which under the right symplectic change of basis corresponds to a projector onto subsystem $A.$ Hence we can set
    \begin{equation}
        \pi_A^C\equiv T
    \end{equation}    
    Now, since this element $\pi_A^C\in\mathfrak t$, by definition of the centralizer this not only means that $\pi_A^C\in\mathfrak e$, but also $\bar{\pi}_A^C\in Z_{\mathfrak t}(K)$ (where by $\bar{\pi}_A^C$ we mean the image of $\pi_A^C$ after removing the nilpotent component)

    Since the only projectors on $Z_{\mathfrak t}(K)$ are those onto the simple components and sums thereof -- since the components are uniquely defined --, there must be some choice of $i$'s such that
    \begin{equation}
        \bar{\pi}_A^C=\sum_{\vec i}\pi_{S_i},
    \end{equation}
    and by lifting this to the full algebra we prove our claim.

\end{proof}

All in all, \Cref{lemma:simple_comp} allows us to focus on a coarser block-diagonalisation, instead of requiring the finest one. Indeed, while each of the simple components of $\mathfrak e$ might contain further projectors leading to a more fine-grained decomposition, that won't give us additional information about potentially pure subsystems, and hence we can limit to analyzing this coarser one, which moreover satisfies the property of being unique.

Computing such a decomposition, can be done efficiently when dealing with $2n\times 2n$ matrices over $\mathbb F_2$. While we won't go into much detail, we will stress the main algorithmic steps for completeness.

\begin{theorem}[{[Efficient computation of simple-component projectors]}]  
    \label{lemma:efficient_SBD}
    There is a deterministic algorithm running in time $\text{poly}(n)$ that computes a complete set of projectors onto the simple components of $Z_\mathfrak t(K).$
\end{theorem}

\begin{proof}
    The full proof can be found stated in full generality in Ref.~\cite{ronyai1990computing}. We will outline the key points here for completeness. We will assume $Z_\mathfrak t(K)$ to be semisimple from the start, since, as we already mentioned, the nilpotent radical can be efficiently handled separately and then projectors from the semisimple part can be lifted to projectors of the full algebra.

    Then, in order to compute this set of projectors we will proceed iteratively in the following way. Given a basis $\{z_i\}_i$ of $Z_\mathfrak t(K)$, we will construct subalgebras
    \begin{equation}
        \mathbb F_2=F_0\subseteq F_1\subseteq\dots\subseteq Z_\mathfrak t(K)
    \end{equation}
    where each $F_{i+1}$ is generated by appending $z_{i+1}$ to $F_i$. Assuming $F_i$ is still a field -- in particular, an extension field of $\mathbb F_2$--, the structure of $F_{i+1}=F_i[z_{i+1]}]$ is determined by the minimal polynomial $f\in F_i[x]$ of $z_{i+1}$ over $F_i$. The minimal polynomial $f$ is defined as the monic polynomial with coefficients in $F_i$ of smallest degree such that $f(z_{i+1})=0.$ For such an $f$ we have that
\begin{equation}
    F_i[z_{i+1]}]\cong F_i[x]/(f),
\end{equation}
which follows almost by definition of the minimal polynomial. Then if $f$ is irreducible over $F_i$, the quotient is again a field: by Bézout's identity, irreducibility of $f$ means that every nonzero class in $F_i[z_{i+1]}]$ is invertible. In that case, no non-trivial idempotent exists yet in $F_{i+1}$. If instead $f$ is reducible, say
\begin{equation}
    f=g\cdot h,\qquad\gcd(g,h)=1
\end{equation}
then
\begin{equation}
    g(z_{i+1})\cdot h(z_{i+1})=0
\end{equation}
and by Bézout's identity, there are $u,v\in F_i[x]$ such that
\begin{equation}
    u\cdot g+v\cdot h=1.
\end{equation}
Evaluating at $z_{i+1}$, we obtain two orthogonal idempotents
\begin{equation}
    e:=v(z_{i+1})h(z_{i+1}),\quad 1-e:=u(z_{i+1})g(z_{i+1}),
\end{equation}
which one can check satisfy
\begin{equation}
    e^2=e,\quad (1-e)^2=1-e,\quad e(1-e)=0.
\end{equation}
These idempotents project onto nontrivial splitting of the center into its simple components. Repeating this procedure inside each block eventually results in the complete decomposition of $Z_{\mathfrak t}(K)$ into its simple components. 

Notably, the complexity of this step relies on polynomial factorisation of the minimal polynomial, which is in general a hard problem. Nevertheless, this complexity heavily depends on the base field of the algebra.  For $\mathbb F_2$, it can be shown that for instance the deterministic version of Berlekamp's algorithm \cite{berlekamp1967factoring} is efficient and runs in time $\text{poly}(n)$.

The output of this algorithm is a set of minimal orthogonal \textit{central} idempotents $\{\pi_{S_j}\}_{j=1}^m$, which project onto orthogonal invariant subspaces of $K$. 

\end{proof}

Note however, that while the projectors found by \Cref{lemma:efficient_SBD} will induce a change of basis that block-diagonalises the symmetries to the degree that matters to us, this change of basis might not be given by a symplectic transformation. In order to make sure that there is a $C\in\text{Sp}(2n,\mathbb F_2)$ that performs the change of basis, it is enough to choose a set of projectors among the ones found, that are self-adjoint with respect to $\Omega$, i.e., defining the adjoint by
\begin{equation}
    \pi^\dagger\equiv\Omega\pi^T\Omega
\end{equation}
we need to impose
\begin{equation}
    \pi^\dagger=\pi.
\end{equation}
\begin{lemma}[{[Symplectic bases from self-adjoint projectors]}]
    \label{lemma:self_adjoint}
    Let $\{\tilde{\pi}_i\}_i$ be some set of orthogonal projectors that satisfy self-adjointness with respect to $\Omega$, and let $\{K_i\}_i$ be their corresponding invariant subspaces. Then there exists a symplectic operation $C\in\text{Sp}(2n,\mathbb F_2)$ that transforms into a basis of $\{K_i\}_i$ and hence block-diagonalises the algebra.
\end{lemma}

\begin{proof}
    Let $\tilde{\pi}_i,\tilde{\pi}_j$ be two such self-adjoint projectors and $K_i=\tilde{\pi}_i(K)$ and $K_j=\tilde{\pi}_j(K)$. Take $u\in K_i$ and $v\in K_j$. Then
    \begin{equation}
       \tilde{\pi}_i(u)=u,\quad \tilde{\pi}_j(v)=v.
    \end{equation}
    Then we can show that $u,v$ commute:
    \begin{align}
    u^T\Omega v&= (\tilde{\pi}_iu)^T\Omega v\\
    &= u^T\tilde{\pi}_i^T\Omega v\notag\\
    &=u^T\Omega\tilde{\pi}_i v\notag\\
    &=u^T\Omega\tilde{\pi}_i\tilde{\pi}_j v\notag\\
    &=0.\notag
    \end{align}

    Since this holds for all projectors, it follows that all subspaces of the corresponding decomposition commute, and thus there is a symplectic transformation $C\in\text{Sp}(2n,\mathbb F_2)$ that maps each $K_i$ to a distinct physical subspace of corresponding dimension.
\end{proof}

Finally, finding such a set of self-adjoint projectors can be done easily thanks to the fact that the \textit{central} projectors $\{\pi_i\}_i$ from \Cref{lemma:efficient_SBD} are unique, since then the self-adjoint projectors $\{\tilde{\pi}_j\}_j$ can be at most sums of them. Indeed, there are only two cases: for some \textit{central} projector $\pi_i$, either it is already self-adjoint, and hence $\pi_i(K)$ admits a symplectic basis, or $\pi_i^\dagger\neq\pi_i$, and then $\tilde{\pi}_i=\pi_i+\pi_i^\dagger$ is the smallest self-adjoint projector containing $\pi_i.$

\begin{lemma}[{[Adjoint closure of the centralizer]}]
    \label{lemma:closed_centralizer}
    The centralizer $\mathfrak e$ is closed under the adjoint, and for any \textit{primitive central} idempotent $\pi_i$ such that $\pi_i^\dagger\neq\pi_i$, $\pi_i^\dagger=\pi_j$, for some other primitive projector $\pi_j$. 
\end{lemma}

\begin{proof}
    Note that the generators of $\mathfrak t$, i.e., $T_i=\Omega B_i$ are themselves self-adjoint.
    \begin{equation}
        T_i^\dagger=\Omega T_i^T\Omega=\Omega B_i^T\Omega\Omega=\Omega B_i=T_i,
    \end{equation}
    since the $B_i$ are symmetric. This then means that $\mathfrak t$ is closed under the adjoint. Take $T=\prod_{i=1}^sT_i$, then
    \begin{equation}
        T^\dagger=\prod_{i=s}^1T_i^\dagger=\prod_{i=s}^1T_i\in\mathfrak t,
    \end{equation}
    and the sum $T_i+T_j$ trivially respects the adjoint. Finally, let $X\in\mathfrak e$. We can show that $X^\dagger$ also commutes with the action of $\mathfrak t$. Let $T\in\mathfrak t$. Since $T^\dagger\in\mathfrak t$,
    \begin{equation}
        XT^\dagger=T^\dagger X.
    \end{equation}
    Taking the adjoints
    \begin{equation}
        TX^\dagger=X^\dagger T,\quad \forall T\in\mathfrak t
    \end{equation}
    and hence $X^\dagger\in\mathfrak e$. Finally, if $\pi\in\mathfrak e$ satisfies $\pi^2=\pi$, then
    \begin{equation}
        (\pi^\dagger)^2=(\Omega\pi^T\Omega)^2=\Omega(\pi^T)^2\Omega=\Omega(\pi^2)^T\Omega=\pi^\dagger,
    \end{equation}
    i.e., the adjoint is also a projector. From the uniqueness of the central idempotents, $\pi^\dagger$  must in fact be either some other central idempotent $\pi'$, or a sum thereof.
\end{proof}

After having performed this lift, we are left with a set of orthogonal self-adjoint projectors onto symplectic invariant subspaces. While these might not yield the finest symplectic block-diagonalisation, that is not important in our setting since, as we already showed, any further block won't follow from a product cut. We summarize the content of these lemmas in the following theorem.

\blockdiagonalisesym*

\begin{proof}
        The proof follows from the lemmas introduced in this subsection. The quadratic forms $\{M_i\}_{i=1}^{d_{\text{Ker}}}$ can be transformed to linear operators by \Cref{lemma:bilinear_to_linear} in order to exploit the structure of finite algebras. By \Cref{lemma:simple_comp} we can focus on finding blocks associated to the \textit{central} projectors, which are unique and can be efficiently computed for $2n\times 2n$ binary matrices by \Cref{lemma:efficient_SBD}. Finally, thanks to their uniqueness, lifting the decomposition to a symplectic one can be done by constructing self-adjoint projectors as shown in \Cref{lemma:self_adjoint} and \Cref{lemma:closed_centralizer}. 

        For a state of the form  $\ket{\psi}=U_C\bigotimes_{i=1}^m\ket{\phi_i}_{c_i}$, the procedure always yields some $\tilde C$, although this need not correspond to  $C$, since in particular we are always free to apply local Cliffords. Nevertheless, $U_{\tilde C}$ will be such that $U_{\tilde C}^\dagger\ket{\psi}$ is guaranteed to be a product across all $c_i$. This follows because the projectors onto the invariant subspaces are unique. Thus, the algorithm will always find the projectors onto the invariant subspaces defined by the $c_i$'s. After the symplectic change of basis under $\tilde C$, these projectors take the form $\mathbb I_{c_i}\oplus 0_{c_i^c}$, or equivalently, after applying $U_{\tilde C}^\dagger$ to the state, all symmetries $\Omega_+^{c_i}$ lie in its quadratic kernel, up to some potential affine shift. By \Cref{lemma:diagonal_stab} this affine shift does not change the purity, and hence the state is effectively disentangled.

\end{proof}

\subsection{Sample complexity analysis II: Polynomial runtime algorithm}

For the purpose of the exponential runtime algorithm, 
we already showed that $O(\text{poly}(n))$ samples are enough to guarantee that with high probability, only the correct transformed subsystem purities will appear in the kernel.

In this case however, unlike before, the algorithm does not just exhaustively search for those symmetries, but rather exploits the rigidity of the full kernel in order to do so. To that end, in order to make sure that the computationally efficient algorithm will find those right purities, we need to guarantee that with an efficient number of samples, all other symmetries in the kernel respect the cuts. We show this in the following Lemma.

\begin{lemma}[{[Finite-sample block diagonalisation]}]
    \label{lemma:block_sample}
    Given a state vector $\ket{\psi}=U_C\ket{\psi_A}\otimes\ket{\psi_B}$
    such that
   
    \begin{equation}
        \max_{S\in\mathcal P_n\backslash \mathbb I}|\langle S\rangle_{\psi}|\leq 1-\eta,
    \end{equation}
    for some $\eta>0.$
  Then, given $N=O\bigg(\frac{1}{\eta^2}\big(n^2+\log(\frac{1}{\delta})\big)\bigg)$ many samples, the observed kernel of $\ket{ \psi}$, $\text{Ker}_{2,N}(\psi)$, is symplectically block-diagonalisable 
  across $A|B$ with probability at least $1-\delta.$
    
\end{lemma}

\begin{proof}
    Given a state of 
    the form 
    \begin{equation}
        \ket{\psi}=U_C\ket{\psi_A}\otimes\ket{\psi_B},
    \end{equation}
    such that 
    \begin{equation}
        \max_{S\in\mathcal P_n\backslash \mathbb I}|\langle S\rangle_{\psi}|\leq 1-\eta,
    \end{equation}
    we will show that the probability of observing a symmetry that is not block-diagonalisable across the cut $A|B$ is exponentially small in the number of samples. This then carries over to the general case where the product is across multiple cuts.

    Just like in the exact case we will first focus on the product state vector $\ket{\psi_A}\otimes\ket{\psi_B}$ itself, since showing that all symmetries of $\ket{\psi}$ are block-diagonalisable by $C$ is equivalent to showing that all symmetries of $\ket{\psi_A}\otimes\ket{\psi_B}$ are block-diagonal. Then we will show that for this state, the probability of observing a non-block-diagonal symmetry $M\in\text{Ker}_{2,N}(\psi_A\otimes\psi_B)$ is upper bounded by the probability of satisfying a linear symmetry. Indeed, let Bell samples be $(q_A,q_B)\in\mathbb F_2^{2|A|}\times \mathbb F_2^{2|B|}$, and take some symmetry 
    \begin{equation}
        M=\begin{pmatrix}
            M_A&M_{A, B}\\ 0& M_B
        \end{pmatrix}
    \end{equation}
    such that $Q_M(q_A,q_B)=(q_A,q_B)^TM(q_A,q_B)=q_A^TM_Aq_A+q_B^TM_Bq_B+q_A^TM_{A, B}q_B$. Moreover, let us define the difference distribution
    \begin{equation}
        \nu_A(u):=\mathsf p[q_A+q_A'=u],\qquad \nu_B(v):=\mathsf p[q_B+q_B'=v].
    \end{equation}
    The probability that $M\in\text{Ker}_{2,N}(\psi_A\otimes\psi_B)$ is then
    \begin{equation}
        s_M:=\mathsf p_{\substack{(q_A,q_B)\sim(\mathsf b_A,\mathsf b_B)\\(q_A',q_B')\sim(\mathsf b_A,\mathsf b_B)}}[Q_M(q_A,q_B)+Q_M(q_A',q_B')=0].
    \end{equation}
    We will also define $\delta_M:=1-s_M$ to be its defect. Again, since the Bell distribution factorises, given two samples $(q_A,q_B),(q_A',q_B')$ with some non-zero probability, the samples $(q_A',q_B),(q_A,q_B')$ also have non-zero probability, so
    \begin{equation}
        Q_M(q_A,q_B)+Q_M(q_A,q_B')+Q_M(q_A',q_B)+Q_M(q_A',q_B')=(q_A+q_A')^TM_{A, B}(q_B+q_B').
    \end{equation}
    Let us call
    \begin{align}
        D_1&=Q_M(q_A,q_B)+Q_M(q_A',q_B')\\
        D_2&=Q_M(q_A,q_B')+Q_M(q_A',q_B),\notag
    \end{align}
    
    which satisfy 
    \begin{equation}
        \mathsf p[D_1=1]=\mathsf p[D_2=1]=\delta_M.
    \end{equation}
    Let us now define
    \begin{equation}
        \eta_{M_{A, B}}:=\mathsf p_{\substack{u\sim\nu_A\\ v\sim\nu_B}}.[u^TM_{A, B}v=1]
    \end{equation}
    By \Cref{eqn:cross_term}, we have
    \begin{align}
        \eta_{M_{A, B}}&=\mathsf p[D_1+D_2=1]\\
        &\leq\mathsf p[D_1=1]+\mathsf p[D_2=1]\notag\\
        &=2\delta_M.\notag
    \end{align}
    Thus
    \begin{equation}
        \delta_M\geq\frac{\eta_{M_{A, B}}}{2}.
    \end{equation}
    We can now show that for product states $\eta_{M_{A, B}}$ is related to the probability of showcasing a linear symmetry. Indeed, let us define the linear symmetry margins on the difference distributions
    \begin{equation}
        \gamma_A:=\min_{l\neq0}\mathsf p_{u\sim \nu_A}[l(u)=1],\qquad \gamma_B:=\min_{l\neq0}\mathsf p_{v\sim \nu_B}[l(v)=1],
    \end{equation}
    such that, if $\gamma_i=0$, the respective distributions on $A$ or $B$ satisfy some linear symmetry, and hence showcase a stabilizer. On the other hand, for large values of $\gamma_i$, no such symmetry is satisfied, even approximately. We can then show that
    \begin{align}
        \eta_{M_{A, B}}&=\mathsf p_{\substack{u\sim \nu_A\\v\sim\nu_B}}[u^TM_{A, B}v=1] \\ &=\sum_{v:M_{A, B}v\neq 0}\mathsf p_{u\sim \nu_A}[u^TM_{A, B}v=1]\mathsf p_{v\sim\nu_B}[v]\notag\\ &\geq \sum_{v:M_{A, B}v\neq 0}\gamma_A\mathsf p_{v\sim\nu_B}[v]\notag\\ &=\gamma_A\mathsf p_{v\sim \nu_B}[M_{A, B}v\neq 0]\geq \gamma_A\gamma_B.\notag
    \end{align}
    Where in the last inequality we have used that
    \begin{align}
        \mathsf p_{v\sim\nu_B}[M_{A, B}v\neq0]&=\mathsf p_{v\sim\nu_B}[\exists a:a^TM_{A, B}v\neq 0]\\
        &=p_{v\sim\nu_B}\left[\bigcup_{a\in\mathbb F_2^{2{|A|}}} a^TM_{A, B}v\neq 0\right]\notag\\
        &\geq \max_{a\in\mathbb F_2^{2|A|}}p_{v\sim\nu_B}[a^TM_{A, B}v\neq 0]\geq\gamma_B.\notag
    \end{align}
    Thus, we have
    \begin{equation}
        \delta_M\geq\frac{\gamma_A\gamma_B}{2}
    \end{equation}
    or equivalently
    \begin{equation}
        \label{eqn:probability_bound}
        s_M\leq 1-\frac{\gamma_A\gamma_B}{2}.
    \end{equation}
  Let us now relate $\gamma_A$ and $\gamma_B$ to the actual Pauli expectation values. From \Cref{lemma:bell_to_Pauli} and \Cref{eqn:stabilizer_symm_1,eqn:stabilizer_symm_2,eqn:stabilizer_symm_3,eqn:stabilizer_symm_4,eqn:stabilizer_sym_5} it follows that
    \begin{equation}
        \langle S\rangle_{\psi_i}^2=\sum_{q_i\in\mathbb F_2^{2|i|}}(-1)^{l_s(q_i)}p_i(q_i),
    \end{equation}
    for $i=A,B$, and $l_s$ the linear Bell constraint defined by Pauli $S.$ For the difference distributions $\nu_i$ we then have
    \begin{equation}
        \sum_{u\in\mathbb F_2^{2|i|}}(-1)^{l_s(u)}\nu_i(u)=\bigg(\sum_{q_i\in\mathbb F_2^{2|i|}}(-1)^{l_s(q_i)}p_i(q_i)\bigg)^2=\langle S\rangle^4_{\psi_i}.
    \end{equation}
    Moreover, since
    \begin{equation}
        \sum_{u\in\mathbb F_2^{2|i|}}(-1)^{l_s(u)}\nu_i(u)=\mathsf p_{u\sim \nu_i}[l_s(u)=0]-\mathsf p_{u\sim \nu_i}[l_s(u)=1]
    \end{equation}
    we have
    \begin{equation}
        \mathsf p_{u\sim \nu_i}[l_s(u)=1]=\frac{1-\langle S\rangle_{\psi_i}^4}{2}.
    \end{equation}
    Then, defining
    \begin{equation}
        \alpha_A:=\max_{t\neq 0}|\langle S_t\rangle_{\psi_A}|,\qquad \alpha_B:=\max_{t\neq 0}|\langle S_t\rangle_{\psi_B}|,
    \end{equation}
    it holds that
    \begin{eqnarray}
        \gamma_A&=&\frac{1-\alpha_A^4}{2}\\
        \gamma_B&=&\frac{1-\alpha_B^4}{2}\notag
    \end{eqnarray}
    Plugging this in \Cref{eqn:probability_bound} we get
    \begin{equation}
        s_M\leq1-\frac{(1-\alpha_A)(1-\alpha_B)}{8}.
    \end{equation}
    Since both $\alpha_A,\alpha_B\leq 1-\eta$, it holds that
    \begin{equation}
        s_M\leq1-\frac{\eta^2}{8}.
    \end{equation}

    Now, let
    \begin{equation}
        p_c:=\mathsf p_{q\sim\mathsf b_\psi}[Q_M(q)=c],\quad c\in\mathbb F_2.
    \end{equation}
    Then, for two independent samples
    \begin{equation}
        s_M=p_0^2+p_1^2=1-2p_0p_1.
    \end{equation}
    Therefore
    \begin{equation}
        p_0p_1=\frac{1-s_M}{2}\geq\frac{\eta^2}{16}.
    \end{equation}
    For a fixed reference sample $q_0,$ let $c_0=Q_M(q_0).$ The probability that one fresh sample is accepted is then $p_{c_0}$. Since
    \begin{equation}
        1-p_{c_0}=p_{1-c_0}\geq p_0p_1,
    \end{equation}
    we have
    \begin{equation}
        p_{c_0}\leq 1-p_0p_1\leq1-\frac{\eta^2}{16}.
    \end{equation}
    Conditioned on $q_0$, after taking $N$ many samples
    \begin{align}
        \mathsf p[M\in\text{Ker}_{2,N}(\psi)|q_0]&=p_{c_0}^N\\
        &\leq\left(1-\frac{\eta^2}{16}\right)^N\notag\\
        &\leq e^{-N\eta^2/16}.
        \notag
    \end{align}
    In total, there are roughly $2^{O(n^2)}$ many wrong symmetries. Doing the corresponding union bound we have
    \begin{equation}
        \mathsf p_N(\text{wrong}\ M)\leq 2^{O(n^2)}\cdot e^{-N\eta^2/16}
    \end{equation}
    Setting a maximum error probability of $\delta$, means it is enough to use 
    \begin{align}
        \label{eqn:samples_block}
        N\geq O\bigg(\frac1{\eta^2}\big(n^2+\log(\frac{1}{\delta})\big)\bigg)
    \end{align}
    many samples to guarantee that all symmetries in the kernel will be block-diagonal across $A|B$, with probability $1-\delta$. Finally, by \cref{cor:block-diagonalisable} this carries over to state vectors of the form $\ket{\psi}=U_C\ket{\psi_A}\otimes\ket{\psi_B}$, and given $N$ many samples as in \Cref{eqn:samples_block} we can guarantee that all symmetries in $\text{Ker}_{2,N}(\psi)$ are block-diagonalisable across $A|B$ with probability $1-\delta$.
    \end{proof}

Let us stress that while we have proven \Cref{lemma:block_sample} for the case of one single partition $A|B$, it follows immediately for the generic case.
Then, given this guarantee on $\text{Ker}_{2,N}(\psi)$, we can show that \Cref{alg:hidden-clifford-test} will manage to solve both tasks in \Cref{def:testing_task,def:learning_task} efficiently.

\maintesting*

\begin{proof}
    Let us analyze first the situation where $\ket{\psi}$ satisfies Case A.
    \Cref{alg:hidden-clifford-test} proceeds by first performing stabilizer learning of the state vector $\ket{\psi}$, which requires $O(\frac{n+\log(1/\delta)}{\eta})$ many samples. If this returns some stabilizer, we can already guarantee that $\ket{\psi}$ is in Case A. Otherwise, for $N=O\big(\frac{1}{\eta^2}\big(n^2+\log(\frac{1}{\delta})\big)\big)$ many samples, $\text{Ker}_{2,N}(\psi)$ is going to be block-diagonalisable across $A|B$ with probability at least $1-\delta$. When this is satisfied, the SBD algorithm will detect it. In the last step, \Cref{alg:hidden-clifford-test} tests which of the cuts returned by SBD correspond to actual pure subsystems, by checking the respective subsystem purity symmetries. Since $\ket{\psi_A}$ and $\ket{\psi_B}$ are true pure subsystems the algorithm tests positive for these.

    If instead $\ket{\psi}$ belongs to Case B, the presence of approximate stabilizers is already ruled out. Then, $\text{Ker}_{2,N}(\psi)$ might still present some spurious structure such that the SBD subroutine returns some symplectic action and some potential cuts. Nevertheless, since the algorithm then ultimately tests for the purities, the analysis of \Cref{lemma:sample_test_exp} carries over. Namely, for $N'=O\big(\frac{1}{\epsilon^2}\big(n^2+\log(\frac{1}{\delta})\big)\big)$, no wrong symmetry $C^T(\Omega_+^S+d)C\in\text{Ker}_{2,N'}(\psi)$ with probability at least $1-\delta$. Thus, regardless of the outcome of the SBD the algorithm tests negative.

    Overall, setting $\lambda=\min(\eta,\epsilon)$ guarantees that with $N=O\big(\frac{1}{\lambda^2}\big(n^2+\log(\frac{1}{\delta})\big)\big)$, the algorithm succeeds in both cases with probability at least $1-\delta.$

\end{proof}

Analogously, we can prove the task in \Cref{def:learning_task} can also be solved efficiently.
\mainlearning*

\begin{proof}
    We will show that \Cref{alg:hidden-clifford-test} manages to do this. Again, the algorithm first learns the exact stabilizers of $\ket\psi$ using $O(\frac{n+\log(1/\delta)}{\eta})$ many samples. With this, we can efficiently find a Clifford $U_{C'}$ that maps a basis of the stabilizer subgroup to $\{Z_i\}_{i=1}^{d_{S}}$, such that $U_{C'}\ket{\psi}=\ket{\psi'}\otimes\ket{0}^{d_S}$, and where $\ket{\psi'}$ now satisfies
    \begin{equation}
        \max_{S\in\mathcal P_n}|\langle S\rangle_{\psi}|\leq 1-\eta.
    \end{equation}
    With this, we can guarantee that given $N=O\big(\frac{1}{\eta^2}(n^2+\log(\frac{1}{\delta}))\big)$ many samples, $\text{Ker}_{2,N}(\psi')$ is block-diagonalisable across the true cuts $c_i$, and moreover, the corresponding purity symmetries will be in $\text{Ker}_{2,N}(\psi')$, so the algorithm identifies them correctly. While the SBD subroutine might detect additional spurious structure in $\text{Ker}_{2,N}(\psi')$, for $N=O\big(\frac{n}{\epsilon^2}(n^2+\log(\frac{1}{\delta}))\big)$, with probability at least $1-\delta$, \Cref{alg:hidden-clifford-test} will only accept it as part of the learnt description if overall $||\ket{\psi}-U_{\tilde C}\bigotimes_{i=1}^{m'}\ket{\phi_i}_{\tilde c_i}||_2\leq \epsilon$. Otherwise, by \Cref{lemma:sample_learn_exp} we could show that with probability at least $1-\delta$ the corresponding subsystem purities would not be present in $\text{Ker}_{2,N}(\psi')$, and hence the algorithm would dismiss these additional effects. 
    
\end{proof}

\begin{remark}[{[Necessity of the stabilizer gap]}]
    \label{rem:stab_gap} The proofs of \Cref{lemma:block_sample,thm:testing_alg,thm:learning_alg} clarify why the stabilizer gap promise in \Cref{def:testing_task,def:learning_task} is necessary. We have shown that only stabilizer symmetries can obstruct the desired block forms. Structurally, stabilizers and product factors are equivalent up to Cliffords: any stabilizer can be mapped to some $Z_i$, thereby isolating a product qubit. Algorithmically, however, we handle each one in separate stages, so the state entering each stage must satisfy the promises in \Cref{def:testing_task,def:learning_task} exactly. Treating an approximate stabilizer $\tilde S$, with $|\langle\tilde S\rangle|_\psi>1-\eta$, as an exact one might introduce only a small error in the final description, but the residual state $\psi'$ after its removal need not be pure anymore. Since our algorithm is not robust, and purity is essential for the symmetry analysis, the gap is therefore required by our present algorithm. We leave the development of a robust version of our algorithm -- or one that handles stabilizers and products simultaneously -- to future work.
\end{remark}

To wrap up, we can finally show that together with Algorithm 4.19 in Ref.~\cite{bakshi2025learning} we can extend our results to a full state learning algorithm. While the complete analysis of the algorithm was only formulated for the product of single qubit states, this can be straightforwardly extended to the qudit case by accounting for dimensionality factors. Taking this into account we show the overall cost of this task.

\learningstates*

\begin{proof}
    First of all, run \Cref{alg:hidden-clifford-test} on $\ket{\psi}$ to precision $\epsilon_{\text{dis}}=\epsilon/2$, with failure probability $\delta_{\text{dis}}=\delta/2.$ By \Cref{thm:learning_alg}, this can be done with $N_{\text{dis}}=O\left(\frac{1}{\lambda^2}(n^2+\log(\frac{1}{\delta_{\text{dis}}}))\right)$, with $\lambda=\min(\eta,\frac{\epsilon}{2\sqrt{n}})$, in time $\text{poly}(n,\frac{1}{\lambda},\log(\frac{1}{\delta_{dis}}))$.

    With probability $1-\delta/2$ this then returns a Clifford $U_{\tilde C}$ and a partition $\{\tilde{c}_i\}$, such that there is a product state satisfying
    \begin{equation}
        \left|\left|U_{\tilde C}^\dagger\ket\psi-\bigotimes_{i=1}^{m'}\ket{\tilde{\psi}_i}_{\tilde c_i}\right|\right|_2\leq\frac{\epsilon}{2}.
    \end{equation}
    Equivalently, in terms of the fidelity 
     \begin{equation}\label{eq:fidelity-bound-1}
        \left|\langle\psi|U_{\tilde C}\bigotimes_{i=1}^{m'}\ket{\tilde{\psi}_i}_{\tilde c_i}\right|^2\geq \left(1-\frac{\epsilon^2}{8}\right)^2\geq 1-\frac{\epsilon^2}{4}.
    \end{equation}
    In order to run Algorithm 4.19, in the high-fidelity promise regime with precision parameter $\epsilon_{\text{prod}}$, we need to guarantee
    \begin{equation}\label{eq:fidelity-bound-2}
        \left|\langle\psi|U_{\tilde C}\bigotimes_{i=1}^{m'}\ket{\tilde{\psi}_i}_{\tilde c_i}\right|^2\geq \frac{5}{6} + \epsilon_{\text{prod}},
    \end{equation}
    which promises to return a product state vector $\ket{\phi}=\bigotimes_i\ket{\phi_i}_{\tilde c_i}$ such that
    \begin{align}
        |\langle\psi|U_{\tilde C}|\phi\rangle|^2 &\geq \left|\langle\psi|U_{\tilde C}\bigotimes_{i=1}^{m'}\ket{\tilde{\psi}_i}_{\tilde c_i}\right|^2 - \epsilon_{\text{prod}}\\
        &\geq 1-\frac{\epsilon^2}{4} - \epsilon_{\text{prod}}.
    \end{align}
    Then, by setting $\epsilon_{\text{prod}} = \frac{\epsilon^2}{4}$, we get
    \begin{equation}
        |\langle\psi|U_{\tilde C}|\phi\rangle|^2 \geq 1 - \frac{\epsilon^2}{2},
    \end{equation}
    and therefore
    \begin{equation}
        ||\ket{\psi}-U_{\tilde C}\ket{\phi}||_2\leq\sqrt{2-2|\langle\psi|\phi\rangle|^2}\leq\epsilon.
    \end{equation}
    Looking back at \Cref{eq:fidelity-bound-1,eq:fidelity-bound-2}, we have that $\epsilon$ must be in $(0,1/\sqrt{3})$ to satisfy the high-fidelity promise regime. Setting $\delta_{\text{prod}}=\delta/2$ then the total failure probability is smaller than $\delta.$ For $\mathcal E_{\text{dis}}$ and $\mathcal E_{\text{prod}}$ the success events of the disentangling and the state learning procedure respectively (and $\mathcal E_i^c$ their failure events), feeding fresh copies to Algorithm 4.19
    \begin{equation}
        \mathsf p(\mathcal E^c_{\text{prod}}|\mathcal E_{\text{dis}})\leq\frac{\delta}{2}.
    \end{equation}
    Therefore
    \begin{equation}
        \mathsf p(\mathcal E_{\text{dis}}\cap\mathcal E_{\text{prod}})\geq 1-\mathsf p(\mathcal E^c_{\text{dis}})- p(\mathcal E^c_{\text{prod}}|\mathcal E_{\text{dis}})\geq1-\frac{\delta}{2}-\frac{\delta}{2}=1-\delta.
    \end{equation}

    Given this choice of parameters, the learning algorithm in Ref.~\cite{bakshi2025learning} requires $N_{\text{prod}}=O\left(\frac{m'd}{\epsilon^4}\log\frac{1}{\delta}+m'd^2\log\frac{m'}{\delta}\right)$ and runs in time $\text{poly}(m',d,\frac{1}{\epsilon},\log\frac{1}{\delta})$, where the only difference with the single qubit complexity stated in the paper is the dimensional factor involved in the tomography steps.
    Thus, the overall sample complexity is given by
    \begin{equation}
        N=N_{\text{dis}}+N_{\text{prod}}=O\left(\frac{1}{\lambda^2}\left(n^2+\log\frac{1}{\delta}\right)+\frac{m'd}{\epsilon^4}\log\frac{1}{\delta}+m'd^2\log\frac{m'}{\delta}\right),
    \end{equation}
     and the total runtime is 
    \begin{equation}
        t=\text{poly}(n,d,\frac{1}{\lambda},\log\frac{1}{\delta}),
    \end{equation}
    where we have used $m'\leq n$.
\end{proof}

\section{Learning \textit{2n} $T$-doped Clifford unitaries}
\label{sec:Tdoped}
We show here a compelling application of our state learning algorithm. We will show that our method gives rise to a proper learning algorithm for learning T-doped Clifford unitaries with up to $2n$ $T$ when some condition is met. In fact we will show that for $\sim n$ many $T$ gates, a value that is significantly higher than the standard threshold of $O(\log(n))$, this condition is satisfied with high probability. In fact, we will also show that in some specific instances the $2n$ value can rise to even higher number of $T$'s. Before moving on however, note that in this section, since we won't need to use symplectic representation, we will refer to the Clifford unitaries by $C$ instead of $U_C$, in order to ease notation.

Like we already mentioned in the introduction, when applied to Choi states, our learning algorithm can be used to learn unitaries of the form
\begin{equation}
    U=C_1\big(\bigotimes_{i=1}^kU_i\big)C_2
\end{equation}

which have a Choi state of the form:
\begin{equation}
    \ket{\Phi_U}=\big(C_1\big(\bigotimes_{i=1}^kU_i\big)C_2\otimes\mathbb I\big)\ket{\Phi_0}=(C_1\otimes C_2^T)\cdot\bigotimes_{i=1}^k\ket{\Phi_{U_i}}.
\end{equation}
While unitaries of this form could already be learnt \cite{bakshi2025learning}, our algorithm requires only access to the Choi state, thus removing the oracle access requirement to $U$ and $U^\dagger$. More importantly, working at the level of the Choi state, allows us to use certain disentangling techniques from \cite{fux2025disentangling}, which can't be used when just working at the unitary level.

Indeed, it was shown in Ref.~\cite{fux2025disentangling}, that certain families of states can be disentangled and efficiently written as the action of a Clifford on an $MPS$, as long as some conditions are fulfilled. In particular, by mapping to an error correcting setting, they show that some T-doped states satisfy these conditions with high probability as long as the number of $T$ gates stays within a feasible regime. Here, we use their techniques to the specific setting of Choi states of T-doped Cliffords, and show both that their rough $n$ T-gate limit can be extended to $2n$ -- as one would naturally expect-- and that for this particular case the disentangling criterion can be phrased in much simpler algebraic terms, without the need of resorting to error correction. Moreover, we will show that this criterion is satisfied with high probability over all such circuits, and show in which instances this can be made even higher.

In the following lemma we prove our own version of the disentangling procedure in the Choi state setting. While some of this could be derived from \cite{fux2025disentangling} we prove it here from scratch for completeness.

\begin{theorem}[{[Disentangling algebraically independent Pauli rotations]}]
    \label{lemma:disentangling}
    Given a unitary of the form 
    \begin{equation}
        U=C \prod_{i=1}^Ke^{i\theta_iP_i}
    \end{equation}
     with the Paulis $\{P_i\}_{i=1}^K$ being algebraically independent, then its Choi state vector $\ket{\Psi_U}=(U\otimes\mathbb I)\ket{\Phi}$ can be written as
    \begin{equation}
        \ket{\Psi_U}=\tilde{C}\bigg(\bigotimes_{i=1}^K\ket{\phi_i}\bigg)\otimes\ket{0}^{\otimes 2n-K}.
    \end{equation}
\end{theorem}
\begin{proof}
    The proof relies on efficiently finding disentanglers that map the multi-qubit Paulis $\{P_i\}_{i=1}^K$ to single qubit Paulis $\{Y_i\}_{i=1}^K$ \textit{wlog}. To that end we show that by exploiting the algebraically independent condition (which can hold up to $K=2N$) one can always find Clifford disentanglers of the form introduced in Ref.~\cite{fux2025disentangling}.

    Assume we have $K=2n$. Then 
    \begin{equation*}
        \langle\{P_i\}_{i=1}^{2n}\rangle=\langle\{X_i,Z_i\}_{i=1}^{n}\rangle, 
    \end{equation*}
    where $\langle\cdot\rangle$ means the algebraic span. On the other hand, the Choi state can be written as
    \begin{equation}
        \ket{\Psi_U}=(C\cdot\prod_{i=1}^{2n}e^{i\theta_iP_i}\otimes\mathbb I)\cdot C_{\Phi}\ket{0}^{\otimes 2n},
    \end{equation}
    where $C_{\Phi}$ is the Clifford preparing the Bell state vector $\ket{\Phi}.$ In the Pauli basis the action of $C_{\Phi}$ is given by
    \begin{align}
        X_i&\mapsto X_i\\
        Z_i&\mapsto Z_i\otimes X_{n+i}.\notag
    \end{align}
    Hence by conjugating the Pauli rotations through $C_{\Phi_0}$, $\ket{\Psi_U}$ becomes
    \begin{equation}
        \ket{\Psi_U}=C\cdot C_{\Phi}\cdot\prod_{i=1}^{2n}e^{i\theta_i\tilde{P_i}}\ket{0}^{\otimes 2n}
    \end{equation}
    where
    \begin{equation}
        \label{eqn:pauli_span}
        \langle\{\tilde{P}_i\}_{i=1}^{2n}\rangle=\langle\{X_i,Z_i\otimes X_{n+i}\}_{i=1}^{n}\rangle.
    \end{equation}
    With this, we now state how to iteratively perform the Clifford disentangling of the Paulis. Take Pauli $\tilde{P}_1$. Given \Cref{eqn:pauli_span}, we know that there is some register $j$ where $\tilde{P}_1^{(j)}=X_j$. If that would not be the case then $\tilde{P}_1=\bigotimes_{i\in S}Z_i\notin \langle\{X_i,Z_i\otimes X_{n+i}\}_{i=1}^{2n}\rangle$. Then $wlog$ we can assume $j=1$ and
    \begin{equation}
        \tilde P_1=X_1\otimes \tilde Q_1
    \end{equation}
   
    Now, it is easy to check that the following holds:
    \begin{equation}
        e^{i\theta_1 X_1\otimes\tilde{Q}_1}=e^{i\pi/4 Z_1\otimes\tilde{Q}_1}\cdot e^{i\theta_1 Y_1\otimes\mathbb I}\cdot e^{-i\pi/4 Z_1\otimes\tilde{Q}_1}
    \end{equation}
    and consequently
    \begin{align}
        \ket{\Psi_U}&=\dots  e^{i\theta_1X_1\otimes\tilde{Q}_1}\ket{0}^{\otimes 2n}\\
        &= \dots e^{i\pi/4 Z_1\otimes\tilde{Q}_1}\cdot e^{i\theta_1 Y_1\otimes\mathbb I}\cdot e^{-i\pi/4 Z_1\otimes\tilde{Q}_1}\ket{0}^{\otimes 2n}\notag\\
        &=\dots e^{i\pi/4 Z_1\otimes\tilde{Q}_1}\cdot e^{i\theta_1 Y_1\otimes\mathbb I}\cdot e^{-i\pi/4 \hspace{0.1cm}\mathbb I_1\otimes\tilde{Q}_1}\ket{0}^{\otimes 2n}\notag\\
        &= \dots e^{i\pi/4 (Z_1-\mathbb I)\otimes\tilde{Q}_1}\cdot e^{i\theta_1 Y_1\otimes\mathbb I}\ket{0}^{\otimes 2n}.\notag
    \end{align}
    Where in the third equality we have used that the disentangler acts on the $\ket{0}^{\otimes 2n}$ state, and $Z\ket{0}=\ket{0}$. In the last equality we then used that $[Z_1\otimes\tilde{Q}_1,\mathbb I_1\otimes\tilde{Q}_1]=0$.
    After these changes, the Choi state becomes
    \begin{equation}
        \ket{\Psi_U}=C\cdot C_{\Phi}\cdot\prod_{i=1}^{2n}e^{i\theta_i\tilde{P_i}}\ket{0}^{\otimes 2n}=C\cdot C_{\Phi}\cdot C^d_{1}\prod_{i=2}^{2n}e^{i\theta_i\tilde{\tilde{P}}_i}\ket{\phi_1}\ket{0}^{\otimes 2n-1}
    \end{equation}
    where $C^d_1=e^{i\pi/4 (Z_1-\mathbb I)\otimes\tilde{Q}_1}$ is another Clifford, and $\tilde{\tilde{P}}_i=C_1^{d\dagger}\tilde{P}_iC_1^d$. We stress that these steps follow the procedure in Ref.~\cite{fux2025disentangling}, but where the specific Clifford $C_\Phi$ has been appended at the end.
    
    Then, in order to see how to bring the Choi state into the desired form we are just left to prove that this procedure can be repeated iteratively. For this to be the case, we must be able to find, at each step, a register $j$ that is still in the $\ket{0}$ state, and where the Pauli to be disentangled does not act trivially with $Z_j$.
    
    In order to see that this is the case, imagine that we disentangled $\tilde{P_1}$, which had Pauli X in register $1.$ After conjugating $\{\tilde P_i\}_{i=2}^{2n}$ through $C_1^d$, they result in $\{\tilde{\tilde P}_i\}_{i=2}^{2n}$. Now, let the set $\{P'_i\}_{i=2}^{2n}$ correspond to the $\tilde{\tilde P}_i$'s, after removing qubit 1, i.e.,
    \begin{equation}
        P_i'=\tilde{\tilde P}_i^2\otimes\dots\otimes \tilde{\tilde P}_i^{2n}.
    \end{equation}
    Then we can show that $\{P'_i\}_{i=2}^{2n}$ are still algebraically independent, and their span is $\langle\{X_{n+1}\}\oplus\{Z_i\otimes X_{n+i},X_i\}_{i=2}^{n}\rangle$.
    Assume they are algebraically dependent, i.e.,
    \begin{equation}
        \prod_{i=2}^{2n}P_i'=\mathbb I.
    \end{equation}
    This could only be the case if one of the two holds
    \begin{align}
        \prod _{i=2}^{2n}\tilde{\tilde P}_i&=\mathbb I\\
        &\text{or}\notag\\
        \prod _{i=2}^{2n}\tilde{\tilde P}_i&=R_1
    \end{align}
    i.e., either their preimage (i.e., the Paulis before removing qubit 1) was already algebraically dependent, or if it was not, the product was $\mathbb I$ in all positions, but qubit 1. Since the first case can't happen by assumption, we are left to prove that
    \begin{equation}
        \label{eqn:prod_q1}\prod_{i=2}^{2n}\tilde{\tilde P}_i=R_1,
    \end{equation}
    also leads to contradiction. However \Cref{eqn:prod_q1} implies that
    \begin{equation}
        \prod_{i=2}^{2n}\tilde{P}_i=C_1^dR_1C_1^{d\dagger}=e^{i\pi/4 (Z_1-\mathbb I)\otimes\tilde{Q}_1}R_1e^{-i\pi/4 (Z_1-\mathbb I)\otimes\tilde{Q}_1}.
    \end{equation}
    
    From this we can show that any choice of $R_1$ leads to contradiction.
    \begin{enumerate}
        \item For $R_1=Z_1$ we would have
        \begin{equation}
        \prod_{i=2}^{2n}\tilde{P}_i=Z_1.
    \end{equation}
    which is not in the algebra spanned by $\{\tilde P_i\}_{i=1}^{2n}$.
    \item For $R_1=X_1$ 
    \begin{equation}
        \prod_{i=2}^{2n}\tilde{P}_i=Y_1\otimes\tilde Q_1.=\tilde P_1\cdot Z_1
    \end{equation}
    which is again a contradiction by the fact that $Z_1$ is not in the algebra spanned by the Paulis.
    \item For $R_1=Y_1$
    \begin{equation}
        \prod_{i=2}^{2n}\tilde{P}_i=X_1\otimes\tilde Q_1.=\tilde P_1
    \end{equation}
    which is a contradiction by the algebraic independence of the Paulis.
    \end{enumerate}
    Hence, we have shown that after conjugating the Paulis through the disentangler $C_1^d$, we can virtually remove qubit 1, and consider the Paulis to be algebraically independent on a $2n-1$ qubit register, where they span $\langle\{X_{n+1}\}\oplus\{Z_i\otimes X_{n+i},X_i\}_{i=2}^{n}\rangle$. Similarly, assume that after this we use qubit $n+1$ to disentangle $\tilde{\tilde P}_2$. Then after repeating the process and removing qubit $n+1$, by the same argument as before, the Paulis remain algebraically independent, and the algebra they span is $\langle\{Z_i\otimes X_{n+i},X_i\}_{i=2}^{n}\rangle$. Thus, since the state on qubits $[2n]\backslash \{1,n+1\}$ is still $\ket{0}^{\otimes 2n-2}$, and the algebra spanned by the remaining Paulis is the same as at the initial step, this procedure can be repeated until every Pauli is disentangled.

    The case where $K<2n$ immediately follows from this. In that case the algebra spanned by the Paulis satisfies
    \begin{equation}
        \langle \{ P_i\}_{i=1}^K\rangle \subset \langle \{ X_i,Z_i\}_{i=1}^n\rangle 
    \end{equation}
    and the state has some stabilizers which can directly be disentangled by mapping them to $Z_i$ on untouched qubits. The proof proceeds as before, removing those qubits and starting the iteration at a later step on the smaller subsystem.
 
\end{proof}

T-doped Clifford unitaries can always be written as $C\prod_{i=1}^Ke^{i\pi/8P_i}$ by moving the $T$ gates through the Cliffords. Then, it follows, that for T-doped Cliffords satisfying the condition in \Cref{lemma:disentangling}, our procedure manages to efficiently learn a full description of its Choi state . Before moving forward and showing how to lift this to a proper learning algorithm, it is useful to notice that in this setting algebraic dependencies in the Paulis limit our disentangling capabilities, and hence play the role of the logical operators as described in Ref.~\cite{fux2025disentangling}. Then, an algebraic dependency of the form
\begin{equation}
    P_{2n+1}=P_1\cdot\dots\cdot P_{2n}
\end{equation}
would be fatal for our algorithm, since this would kill all product structure across all cuts (even if the entanglement it creates is bounded). However, as long as the algebraic dependencies stay bounded in $O(\log(n))$-sized sectors, our algorithm still works efficiently even for $T$-counts higher than $2n$. Indeed, if 
\begin{equation}
    P_{2n+1}=\prod_{i=1}^{r}P_i
\end{equation}
with $r=O(\log(n))$ then
\begin{equation}
    \ket{\Phi_U}=\big(C\prod_{i=1}^{2n+1}e^{i\pi/8P_i}\otimes\mathbb I\big)\ket{\Phi_0}=\tilde C\ket{\phi_R}\bigotimes_{i=r+1}^n\ket{\phi_i}
\end{equation}
with $\ket{\phi_R}$ supported on $r=O(\log(n))$ qubits. Thus, for such an instance, our algorithm could still enable efficient learning of the state.

 Having said this, for the $2n$ algebraic independent setting, we will show how to lift the Choi state learning algorithm to a proper learning algorithm that recovers a description of $U$ in terms of an $n$-qubit Clifford, and $2n$ Pauli rotations. 

\begin{lemma}[{[Proper recovery from a learned Choi state]}]
\label{lemma:properlearning}    
 Given a description of the Choi state of a unitary $U=C\prod_{i=1}^Ke^{i\theta_iP_i}$ as in \Cref{lemma:disentangling}, i.e., in terms of some $2n$-qubit Clifford $\tilde{C}$ acting on some state vector $\ket{\psi}=\bigotimes_{i=1}^{2n}\ket{\phi_i},$ with $\ket{\phi_i}=e^{i\theta_i\tilde R_i}\ket{\sigma_i}$, whenever the $\theta_i$'s are such that the conditions in \Cref{def:learning_task} are satisfied, a proper description of $U$ in terms of $(C,\{\theta_i,P_i\}_{i=1}^K)$ (up to local gauge freedom) can be efficiently recovered.

\end{lemma}
\begin{proof}
    Keep in mind that the $2n$-qubit Clifford that we would learn from the description in \Cref{lemma:disentangling} is of the form $\tilde{C}=\left(C\cdot C_{\Phi_0}\cdot\prod_{i=1}^{2n}C^d_i\right)\cdot\bigotimes_{j=1}^{2n}C_j^{loc}$, where the disentangling Cliffords are $C^d_i=e^{i\pi/4\cdot Z_i\otimes Q_i}\cdot e^{-i\pi/4\cdot\mathbb I_i\otimes Q_i}$, and where $\bigotimes_{j=1}^{2n}C_j^{loc}$ accounts for the local Clifford gauge freedom in our learning algorithm.

    By construction we have that
    \begin{equation}
        \bigl(\prod_{j=1}^{i}C^{d\dagger}_j\bigr)\cdot C_{\Phi_0}^\dagger\cdot C^\dagger P_i C\cdot C_{\Phi_0}\cdot\bigl(\prod_{j=1}^{i}C^d_j\bigr)=R_i,
    \end{equation}
    so that
    \begin{equation}
        \label{eqn:disentangler_Evol}
        \tilde{C}^\dagger P_i\tilde{C}=\left(\bigotimes_{j=1}^{2n}C_j^{loc}\right)^\dagger\bigl(\prod_{j=i+1}^{2n}C^{d\dagger}_j\bigr)R_i\bigl(\prod_{j=i+1}^{2n}C^d_j\bigr)\left(\bigotimes_{j=1}^{2n}C_j^{loc}\right)=\tilde R_i\cdot 
        K_{i+1}^{q_{i+1}(i)}\otimes\dots\otimes K_{2n}^{q_{2n}(i)}
    \end{equation}
    where the $K_i$'s are of the form ${C_i^{loc}}^\dagger Z_i C_i^{loc}$ and the $q_j(i)$'s are some convoluted functions of commutators between Paulis in different registers. Indeed, this follows since any Pauli $W$ with trivial support on qubit $j$ (which is the case in \Cref{eqn:disentangler_Evol}) satisfies $C_j^{d\dagger}WC_j^d=W$ if $[T,Q_j]=0$ and $C_j^{d\dagger}WC_j^d=W\cdot Z_j$ if $\{T,Q_j\}=0$, so that after going through all the disentanglers, $R_i$ will pick up some $Z$ terms according to the commutation relations between Paulis. These $Z$ terms are then mapped to some other Pauli on qubit $i$ by the local Cliffords.
    
    Hence, we have that
   
    \begin{equation}
        \label{eqn:single_qubit_preimage}
        \tilde C\tilde R_i\tilde C^\dagger = P_i\cdot \prod_{j=i+1}^{2n}\tilde{C}K_j^{q_j(i)}\tilde C^\dagger.
    \end{equation}
    We can now check that, while all $P_i$'s are supported on qubits $1\dots n$, $\tilde C K_j\tilde C^\dagger$ need to have non-trivial support on qubits $n+1\dots 2n$. Assume there is some $K_A:=\prod_{a\in A}K_a$ that is supported only on qubits $1\dots n$ after conjugating with $\tilde C$. Since $\{P_i\}_{i=1}^{2n}$ generates the full Pauli group on $n$-qubits we could write 
    \begin{equation}
        \tilde CK_A\tilde C^\dagger=\prod_{b\in B}P_b,
    \end{equation}
    for some $S\subseteq[2n]$. However, then
    \begin{equation}
        \label{eqn:K_algebraic_dependence}
        K_A= \prod_{b\in B}\tilde C^\dagger P_b\tilde C = \prod_{b\in B}\tilde R_b\cdot (K_{b+1}^{q_{b+1}(b)}\otimes\dots\otimes K_{2n}^{q_{2n}(b)})=\prod_{a\in A}K_a
    \end{equation}
    and thus
    \begin{equation}
        \prod_{b\in B}\tilde R_b=\prod_{a\in A}K_a\cdot\tilde K,
    \end{equation}
    where $\tilde K$ is the product of $K_j$ terms in \Cref{eqn:K_algebraic_dependence}. Nevertheless, this can't be the case, since the operators in the set $\{\tilde R_i,K_i\}_{i=1}^{2n}$ are algebraically independent.

    With this observation it is then clear how to retrieve the set $\{P_i\}_{i=1}^{2n}$ in this case. Namely, after running our learning algorithm the single qubit states are of the form $\ket{\phi_i}=e^{i\theta_i\tilde R_i}\ket{\sigma_i},$ with $\ket{\sigma_i}=K_i\ket{\sigma_i}.$ Performing single qubit state tomography can then recover $\{\tilde R_i,K_i\}_{i=1}^{2n}.$ Then, after computing the preimage of these single qubit Paulis under $\tilde C$, we can obtain the $P_i$'s from \Cref{eqn:single_qubit_preimage} by imposing that their support on qubits $n+1\dots 2n$ needs to be trivial, which can be done by solving a definite system of $2n$ linear equations.

    On the other hand, in order to retrieve the $n$-qubit Clifford we can do as follows. Having learnt the $2n-$qubit Clifford $\tilde C$ and $\ket{\sigma_i},$ with their corresponding $K_i$'s, we can construct the state vector $\ket{\Gamma}=\tilde C\bigotimes_{i=1}^{2n}\ket{\sigma_i}$, which is a stabilizer state with stabilizers of the form $S_i=\tilde CK_i\tilde C^{\dagger}.$ This state can now be shown to be the Choi state of an $n$-qubit Clifford. Indeed, assume $\ket{\Gamma}$ is not maximally entangled across the input-output cut. Since $\ket{\Gamma}$ is a stabilizer, there would be a nonidentity stabilizer supported only on the output registers:
    \begin{equation}
        \prod_{i\in T}S_i=V\otimes \mathbb I
    \end{equation}
    for some nonempty $T$ and $V\neq\mathbb I$. However, its expectation on the learnt state $\tilde C\ket{\psi}$ is
    \begin{align}
        \bra{\psi}\tilde C^\dagger V\otimes\mathbb I\tilde C\ket{\psi}&=\bra{\psi}\tilde C^\dagger \prod_{i\in T}S_i\tilde C\ket{\psi}\\
        &=\prod_{i\in T}\bra{\phi_i}K_i\ket{\phi_i}\notag\\
        &\neq0\notag.
    \end{align}
    But since $\tilde C\ket{\psi}$ is the Choi state of a unitary and it is fully entangled across the input-output cut, its marginals are maximally mixed. This implies
    \begin{equation}
        \bra{\psi}\tilde C^\dagger V\otimes\mathbb I\tilde C\ket{\psi}=0,
    \end{equation}
    which is a contradiction. Thus, $\ket{\Gamma}$ is maximally entangled across input-output qubits. Since it is also a stabilizer state,
    \begin{equation}
        \ket{\Gamma}=C'\otimes\mathbb I\ket{\Phi_0},
    \end{equation}
    where $C'$ is equivalent to $C$ up to local Clifford action, and hence we have recovered the proper description up to the gauge freedom. Again, while we have proved this for $K=2n$, for $K<2n$, the learnt Choi state will showcase some stabilizers which can be mapped to distinct registers, and just like before, the argument proceeds in the same way.

\end{proof}

Finally, we show that for $T$-doped Clifford circuits with $2n$ many $T$ gates, the condition that the corresponding Paulis are algebraically independent is satisfied with high probability, and hence our algorithm works for most instances in this regime.
\independentpauliprobabilities*

\begin{proof}
    The proof is straightforward. Since the Clifford gates are assumed independent, the corresponding Paulis are independent and uniformly distributed over $\mathbb F_2^{2n}\backslash\{0\}$. Suppose the first $i-1$ Paulis are independent. Their span contains $2^{i-1}$ vectors, of which $2^{i-1}-1$ are nonzero. Therefore, conditioned on the previous labels being independent, the probability that the $i$-th label is not in their span is
    \begin{equation}
        p_i=1-\frac{2^{i-1}-1}{4^n-1}=\frac{4^n-2^{i-1}}{4^n-1}.
    \end{equation}
    Then, for $k\leq 2n$, the probability that all $k$ Paulis are independent is
    \begin{equation}
        p_{n,k}=\prod_{j=0}^{k-1}\left(1-\frac{2^j-1}{4^n-1}\right).
    \end{equation}
\end{proof}

\probabilities*

\begin{proof}
    The proof follows directly by plugging the corresponding $T$-counts in \Cref{eqn:success_p}
\end{proof}

Note that our proofs worked not just for $T$ gate doping, but arbitrary Pauli rotations with some angle $\theta$ and hence the algorithm from \Cref{cor:state_learning} also works for learning efficiently more generic doped circuits, in particular when $\{\theta_i\}$ are such that the conditions in \Cref{def:learning_task} are satisfied.\\

{\it Acknowledgements.} We thank Jose Carrasco, Salvatore F.E. Oliviero, Louis Schatzki and Boren Gu for useful discussions. This work has been supported by the BMBF (PhoQuant, QPIC, Hybrid++, DAQC, QSolid), the BMWK (EniQmA), the DFG (CRC 183), the Quantum Flagship (PasQuans2, Millenion), the Munich Quantum Valley, Berlin Quantum, and the European Research Council.

\nocite{apsrev42Control}
\bibliographystyle{apsrev4-2}

\bibliography{biblio}

\end{document}